\documentclass[11pt, draftcls, onecolumn]{IEEEtran}
\usepackage{subfigure}
\usepackage{setspace}
\usepackage{amsmath}
\usepackage{amssymb}
\usepackage{amsfonts}
\usepackage{amscd}
\usepackage{mathdots}
\usepackage{mathrsfs}
\usepackage[final]{graphicx}
\usepackage{graphicx}
\usepackage{psfrag}
\usepackage{epsfig}
\usepackage{color}
\usepackage{url}
\usepackage{textcomp}
\usepackage{multirow}
\usepackage{threeparttable}
\input{epsf.sty}
\newtheorem{theorem}{Theorem}
\newtheorem{lemma}{Lemma}
\newtheorem{definition}{Definition}

\begin{document}

\title{A Low ML-decoding Complexity, Full-diversity, Full-rate MIMO Precoder}
\author{
\authorblockN{K. Pavan Srinath and B. Sundar Rajan,\\}
\authorblockA{Dept of ECE, Indian Institute of science, \\
Bangalore 560012, India\\
Email:\{pavan,bsrajan\}@ece.iisc.ernet.in\\
}
}
\maketitle
\vspace{-15mm}
\begin{abstract}
Precoding for multiple-input, multiple-output (MIMO) antenna systems is considered with perfect channel knowledge available at both the transmitter and the receiver. For 2 transmit antennas and QAM constellations, an approximately optimal (with respect to the minimum Euclidean distance between points in the received signal space) real-valued precoder based on the singular value decomposition (SVD) of the channel is proposed, and it is shown to offer a maximum-likelihood (ML)-decoding complexity of $\mathcal{O}(\sqrt{M})$ for square $M$-QAM. The proposed precoder is obtainable easily for arbitrary QAM constellations, unlike the known complex-valued optimal precoder by Collin et al. for 2 transmit antennas, which is in existence for $4$-QAM alone with an ML-decoding complexity of $\mathcal{O}(M\sqrt{M})$ ($M=4$) and is extremely hard to obtain for larger QAM constellations. The proposed precoder's loss in error performance for 4-QAM in comparison with the complex-valued optimal precoder is only marginal. Our precoding scheme is extended to higher number of transmit antennas on the lines of the E-$d_{min}$ precoder for $4$-QAM by Vrigneau et al. which is an extension of the complex-valued optimal precoder for $4$-QAM. Compared with the recently proposed $X-$ and $Y-$precoders, the error performance of our precoder is significantly better. It is shown that our precoder provides full-diversity for QAM constellations and this is supported by simulation plots of the word error probability for $2\times2$, $4\times4$ and $8\times8$ systems.
\end{abstract}

\begin{keywords}
Diversity gain, low ML-decoding complexity, MIMO precoders, singular values, word error probability.
\end{keywords}

\section{Introduction and Background}\label{sec_intro}
Multiple-input, multiple-output (MIMO) antenna systems have evoked a lot of research interest primarily because of the enhanced capacity they provide, compared with that provided by the single antenna point to point channel. Moreover, for a system with $n_t$ transmit antennas and $n_r$ receive antennas ($n_t \times n_r$ system), the maximum {\it diversity gain} (refer Section \ref{sec_model} for a definition of diversity gain) achievable with coherent detection has been shown to be $n_tn_r$. For MIMO systems with the channel state information available only at the receiver (CSIR), suitably designed space-time block codes (STBCs) \cite{seshadri} provide full-diversity. Full-rate transmission is said to occur if $n_{min} = \min(n_t,n_r)$ independent information symbols are transmitted in every channel use. Full-rate STBCs achieving full-diversity have also been proposed \cite{SRS}, \cite{ORBV}. However, all full-rate, full-diversity STBCs are characterized by a high ML-decoding complexity (refer Section \ref{sec_model} for a formal definition of ML-decoding complexity). In general, decoding full-rate STBCs requires jointly decoding $n_tn_{min}$ symbols. 

MIMO systems with full channel state information at the transmitter (CSIT) or partial CSIT have been extensively studied in literature. From an information-theoretic perspective, capacity is an important parameter for MIMO systems and waterfilling \cite{cover} can be employed to achieve the capacity with a Gaussian codebook. From a signal processing point of view, the error performance of MIMO systems using finite constellations is one of the important parameters, and several precoding\footnote[1]{precoding is also referred to as ``transmit beamforming''.} schemes have been proposed in this regard. Maximal ratio transmission was introduced in \cite{mrt} to achieve full-diversity while maximizing the signal-to-noise ratio (SNR) by precoding at the transmitter and equalizing at the receiver for transmission of a single symbol per channel use. Subsequently, the use of precoding and equalizing matrices at the transmitter and the receiver, respectively, was proposed in \cite{max_snr} to maximize the SNR at the receiver, but this scheme resulted in low-rate transmission. Several works on optimal linear precoders and decoders have been done for the {\it minimum mean square error} (MMSE) criterion \cite{gen_linear}-\cite{perez}. Since these precoders are linear and optimal for the MMSE decoding, the decoding complexity is very low and full-diversity is also achieved, but the error performance is worse than that for the ML-decoding. Other non-ML-decoding techniques include lattice-reduction based techniques \cite{lattice_reduction} which provide full-rate transmission with possibly full-diversity, but lattice-reduction itself involves a high complexity for large MIMO systems. Extensive research has also been done on MIMO systems with limited feedback to the transmitter about the channel from the receiver (see, for example, \cite{multimode} and references therein). In this paper, we consider MIMO systems with full CSIT. The channel state information could be either sent to the transmitter by the receiver (when there are separate frequency bands for uplink and downlink transmission) or the transmitter could estimate the channel, if it is reciprocal (like in a time division duplexing (TDD) system), by receiving pilot signals from the receiver. In literature, to the best of our knowledge, there is no known precoding technique to achieve all the three attributes - full-rate, full-diversity and low ML-decoding complexity (``low ML-decoding complexity'' is a relative term and in this paper, it is used to mean the joint decoding of at most 2 complex symbols).

 Almost all the popular precoding techniques with ML-decoding at the receiver use the singular value decomposition (SVD) of the MIMO channel \cite{rayleigh}. The E-$d_{min}$ precoder for $4$-QAM \cite{edmin_precoder}, an extension of the complex-valued optimal\footnote[2]{Throughout this paper, unless otherwise stated, optimalily is with respect to the minimum Euclidean distance between points in the received signal space.} precoder \cite{optimal_precoder} to higher number of transmit antennas, has been shown to perform very well for $4$-QAM, beating all other linear precoding and decoding schemes based on the MMSE criterion, and ML-decoding involves jointly decoding two complex symbols only. However, this precoder exists in literature for 4-QAM alone and is very hard to obtain for larger QAM constellations, since it involves a numerical search over 3 parameters. Recently, $X$- and $Y$- precoders have been proposed in \cite{xy_codes} as rivals for the E-$d_{min}$ precoder. The $X$-precoder has been shown to offer an ML-decoding complexity of $\mathcal{O}(M)$ (this can be brought down to $\mathcal{O}(\sqrt{M})$ by the same decoding scheme as for our precoder, which is explained in Subsection \ref{subsec_compexity}), while the $Y$-precoder has an ML-decoding complexity which is invariant with respect to the constellation size $M$. The disadvantage with the $X$-precoder is that it loses out to the E-$d_{min}$ precoder in error performance for $4$-QAM and it is not known if an explicit expression for the precoding matrix can be obtained for larger QAM constellations. The $Y$-precoder (which uses a two-dimensional constellation), although explicitly obtainable for constellations of any size $M$, loses out in error performance to the E-$d_{min}$ precoder, since it has not been optimized for error performance. In literature, all the aforementioned low ML-decoding complexity precoders have been claimed to offer a diversity gain of $(n_t -n_{min}/2 +1)(n_r -n_{min}/2 +1)$ by the authors (but the simulation results in this paper indicate that the E-$d_{min}$ precoder has full-diversity for $4$-QAM). Concerned by the limitations of each of the low ML-decoding complexity precoders, we first propose a real-valued, approximately optimal precoder (we explain in Section \ref{sec_proposed} why the precoder is ``approximately optimal'') based on the SVD of the channel for $n_t = 2$ and then extend it to higher number of transmit antennas, an approach similar to that in \cite{edmin_precoder}. The ML-decoding complexity offered by our precoder is shown to be $\mathcal{O}(\sqrt{M})$ for $M$-QAM. For $4$-QAM, the proposed precoder has only a marginally poorer error performance than the E-$d_{min}$ precoder, but has lower ML-decoding complexity. For larger QAM constellations, it is easily obtainable, unlike the E-$d_{min}$ precoder. When compared with the $X$- and $Y$-precoders, it has a much better error performance. The main contributions of the paper are - 
\begin{enumerate}
 \item we propose a novel scheme to obtain an SVD-based, real-valued, approximately optimal precoder for $2$ transmit antennas and any $M$-QAM. The method of obtaining this precoder is different from the one taken to obtain the complex-valued optimal precoder for 2 transmit antennas \cite{optimal_precoder}, and is easily applicable for any $M$-QAM, unlike that in \cite{optimal_precoder}.
\item We extend this real-valued precoder to higher number of transmit antennas and show that our precoding scheme offers full-diversity with ML-decoding. This is a new result as the existing low ML-decoding complexity precoders have been claimed to offer a diversity gain of only $(n_t -n_{min}/2 +1)(n_r -n_{min}/2 +1)$. The simulation plots of the word error probability for $2 \times 2$, $4\times4$ and $8\times8$ systems support our claims about full-diversity.
\item The ML-decoding complexity of the proposed precoder is shown to be $\mathcal{O}(\sqrt{M})$ for square $M$-QAM, in general. However, for a considerable number of channel realizations, no search is required over the $M$ signal points. Specifically for $4$-QAM and 2 transmit antennas, simulations reveal that for more than $50\%$ of the channel realizations, no search is needed over any of the signal points. This aspect is elaborated in Subsection \ref{subsec_compexity}.
\end{enumerate}

The rest of the paper is organized as follows. Section \ref{sec_model} gives the system model, the relevant definitions and some known results which are needed for our precoder design. A brief review of existing low ML-decoding complexity precoders is given in Section \ref{sec_low} The method to obtain the proposed precoder is presented in Section \ref{sec_proposed} and its ML-decoding complexity is analyzed in Subsection \ref{subsec_compexity}. In Section \ref{sec_extension}, we show how this precoding scheme can be extended to higher number of transmit antennas while Section \ref{sec_div_gain} deals with the achievable diversity gain with the proposed precoder. Simulation results are given in Secion \ref{sec_simulations} and concluding remarks constitute Section \ref{sec_discussion}.
 
\textit{Notations}: Throughout, bold, lowercase letters are used to denote vectors and bold, uppercase letters are used to denote matrices. For a complex matrix $\textbf{X}$, the Hermitian, the transpose and the Frobenius norm of $\textbf{X}$ are denoted by $\textbf{X}^{H}$, $\textbf{X}^{T}$ and $\Vert \textbf{X} \Vert$, respectively. The $i^{th}$ element of a vector $\textbf{x}$ is denoted by $[\textbf{x}]_i$, the $(i,j)^{th}$ entry of $\textbf{X}$ is denoted by $\textbf{X}(i,j)$, $tr(\textbf{X})$ denotes the trace of $\textbf{X}$, and $\textbf{X} =\operatorname{diag}(x_1,x_2,\cdots,x_n)$ implies that $\textbf{X}$ is a diagonal matrix with $x_1,x_2,\cdots,x_n$ as the diagonal entries. The set of all real numbers, complex numbers and integers are denoted by $\mathbb{R}$, $\mathbb{C}$ and $\mathbb{Z}$, respectively. The real and the imaginary part of a complex-valued vector $\textbf{x}$ are denoted by $\textbf{x}_I$ and $\textbf{x}_Q$, respectively, $\vert x\vert$ denotes the absolute value of a complex number $x$ and $\vert \mathcal{S}\vert$ denotes the cardinality of the set $\mathcal{S}$. The $T\times T$ identity matrix and the $n \times m$ sized null matrix are denoted by $\textbf{I}_T$ and $\textbf{O}_{n \times m}$, respectively. For a complex random variable $X$, $\mathbb{E}[X]$ denotes the expectation of $X$, while $X \sim \mathcal{N}_{\mathbb{C}}\left(0,1\right)$ implies that $X$ has the complex normal distribution with zero mean and unit variance. Unless used as a subscript or to denote indices, $j$ represents $\sqrt{-1}$ and for a function $f(x)$, $\underset{x}{\operatorname{argmin}}f(x)$ and $\underset{x}{\operatorname{argmax}}f(x)$ denote that value of $x$ which minimizes and maximizes $f(x)$, respectively. For any real number $m$, $\lfloor m\rfloor$ denotes the largest integer smaller than $m$, $\lceil m\rceil$ denotes the smallest integer larger than $m$, $\operatorname{rnd}[m]$ denotes the operation that rounds off $m$ to the nearest integer and $\operatorname{sgn}(m)$ gives the sign of $m$, both of which can be expressed as 
\begin{equation*}
\operatorname{rnd}[m] = \left\{ \begin{array}{ll}
\lfloor m \rfloor,  & \textrm{if} ~ \lceil m \rceil - m > m - \lfloor m \rfloor\\
\lceil m \rceil,  & \textrm{otherwise}\\
\end{array} \right., ~~~~~~~~
\operatorname{sgn}(m) = \left\{ \begin{array}{rl}
1,  & \textrm{if} ~ m \geq 0\\
-1, & \textrm{otherwise}.\\
\end{array} \right.
\end{equation*}
The Gamma function and the Q-function of $x$ are denoted by $\Gamma(x)$ and $Q(x)$, respectively, and given as 
\begin{eqnarray*}
 \Gamma(x)  =  \int_{0}^{\infty}e^{-t}t^{x-1}dt,~~~~~~
Q(x)  =  \int_{x}^{\infty}\frac{1}{\sqrt{2\pi}}e^{-\frac{t^2}{2}}dt.
\end{eqnarray*}

Let $f(x)$ and $g(x)$ be two functions. Then, $f(x)= \mathcal{O}\left(g\left(x\right)\right)$ if and only if there exists a positive constant $c < \infty$ such that 
\begin{equation*}
 \lim_{x \to \infty}\frac{f(x)}{g(x)} = c,
\end{equation*}

\noindent and $f(x) = o\left(g(x)\right)$ as $x \to a$ if and only if
\begin{equation*}
 \lim_{x \to a}\frac{f(x)}{g(x)} = 0.
\end{equation*}
\noindent For a real variable $t$, the unit step function $u(t)$ is defined as $u(t) = 1$, if $t >0$, and $u(t) =0$, if $t < 0$.
 
\section{System Model}\label{sec_model}
We consider an $n_t \times n_r$ MIMO system with full CSIT and CSIR. The channel is assumed to be quasi-static and flat with Rayleigh fading. The channel is modelled as 
\begin{equation}\label{channel_model}
\textbf{y} = \sqrt{\frac{SNR}{n_t}}\textbf{Hs} + \textbf{n},
\end{equation}
where $\textbf{y} \in \mathbb{C}^{n_r\times 1}$ is the received vector, $\textbf{H} \in \mathbb{C}^{n_r\times n_t}$ is the channel matrix, $\textbf{s} \in \mathbb{C}^{n_t\times 1}$ is the precoded symbol vector and $\textbf{n} \in \mathbb{C}^{n_r\times 1}$ is the noise vector. The entries of $\textbf{H}$ and $\textbf{n}$ are i.i.d. circularly symmetric complex Gaussian random variables with zero mean and variance 0.5 per real dimension. In \eqref{channel_model}, the scalar $SNR$ is the average SNR at each receive antenna, and $\textbf{s}$ is constrained such that $\mathbb{E}[tr(\textbf{s}\textbf{s}^H)] = n_t$. The precoded symbol vector $\textbf{s}$ can be defined as 
\begin{equation*}
 \textbf{s} \triangleq \frac{1}{\sqrt{E}}\textbf{Mx},
\end{equation*}
where $\textbf{M} \in \mathbb{C}^{n_t \times n_{min}}$ is the precoding matrix, with $\Vert \textbf{M} \Vert^2 = n_t$, and $\textbf{x} \triangleq [x_1,x_2,$ $\cdots, x_{n_{min}}]^T$  is the symbol vector, with its entries taking values independently from a signal constellation denoted by $\mathcal{A}$, having an average energy of $E$ units. The rate of transmission is $n_{min}$ independent symbols per channel use. Note that in this model, the variable scalar which defines the average SNR at each receive antenna is $SNR$, while $E$ is a constant. For example, for a standard $M$-QAM, with $M = 2^{2a}$ for some positive integer $a$, $E = 2(M-1)/3$.

Let $\textbf{H} = \textbf{UDV}^H$, obtained on the SVD of $\textbf{H}$, with $\textbf{U} \in \mathbb{C}^{n_r \times n_r}$ and $\textbf{V} \in \mathbb{C}^{n_t \times n_t}$ being unitary matrices.  $\textbf{D} \in \mathbb{R}^{n_r \times n_t}$ is such that $\textbf{D} = [\textbf{D}_1 ~~ \textbf{O}_{n_r \times (n_t-n_r)}]$ if $n_t \geq n_r$ and $\textbf{D} = [\textbf{D}_1 ~~ \textbf{O}_{n_t \times (n_r-n_t)}]^T$ if $n_t < n_r$, where $\textbf{D}_1 \in \mathbb{R}^{n_{min} \times n_{min}}$ is a diagonal matrix given by $\textbf{D}_1 = \operatorname{diag}(\sigma_1,\sigma_2,\cdots, \sigma_{n_{min}})$, with $\sigma_1,\sigma_2$, $\cdots$, $\sigma_{n_{min}}$ being the non-zero singular values of $\textbf{H}$, placed in the descending order on the diagonal. Let the precoding matrix $\textbf{M}$ be given as
\begin{equation}\label{precoder}
 \textbf{M} = \textbf{VP},
\end{equation}
where $\textbf{P} \in \mathbb{C}^{n_t \times n_{min}}$. Now, \eqref{channel_model} can be written as 
\begin{equation}\label{revised_channel_model}
 \textbf{y}^\prime = \sqrt{\frac{SNR}{n_tE}}\textbf{DPx} + \textbf{n}^\prime,
\end{equation}
where $\textbf{y}^\prime = \textbf{U}^H\textbf{y}$ and $\textbf{n}^\prime = \textbf{U}^H\textbf{n}$, with the distribution of $\textbf{n}^\prime$ being the same as that of $\textbf{n}$.

The ML-decoding rule seeks to find that $\check{\textbf{x}} \in \mathcal{A}^{n_{min} \times 1}$ which minimizes the metric given by
\begin{equation}\label{ml_metric}
 m(\textbf{x}) =  \left \Vert \textbf{y}^\prime - \sqrt{\frac{SNR}{n_tE}}\textbf{DPx} \right \Vert^2.
\end{equation}
Clearly, the error performance of the system depends on the choice of $\textbf{P}$ and $\mathcal{A}$. From \eqref{precoder}, it is evident that the design of the precoding matrix $\textbf{M}$ amounts to designing $\textbf{P}$. Henceforth in this paper, $\textbf{P}$ is referred to as {\it precoder} and the constellation is  assumed to be an $M$-QAM, where $M = 2^{2a}$ for some positive integer $a$.

\begin{definition} ({\it Full-diversity precoder})
In a MIMO system, if at a high SNR, the average probability $P_e$ that a transmitted symbol vector is wrongly decoded is given by
\begin{equation*}
 P_e \approx (G_c.SNR)^{-G_d},
\end{equation*}
where $\approx$ stands for ``is approximately equal to'', then, $G_d$ and $G_c$ are called the {\it diversity gain} (or diversity order) and the {\it coding gain} of the system, respectively. For a MIMO system with precoding, if $G_d = n_tn_r$, then, we call the precoder a full-diversity precoder.
\end{definition}

\begin{definition}\label{def1}({\it ML-Decoding complexity})
The ML decoding complexity is measured in terms of the number of computations involved in minimizing the ML-decoding metric given in \eqref{ml_metric} and is a function of the constellation size $M$. If at most $k$ symbols are required to be jointly decoded, the ML-decoding complexity is said to be $\mathcal{O}(M^{k})$. 
\end{definition}

Note that the above definition of the ML-decoding complexity is with respect to the {\it worst-case} ML-decoding complexity. The use of a sphere decoder \cite{viterbo} can effectively result in a much lower {\it average} ML-decoding complexity that depends on the dimension of the sphere decoder and not on the constellation size \cite{hassibi}. For a complex lattice constellation of size $M$, if the ML-decoding complexity is $\mathcal{O}\left(M^k\right)$, the dimension of the real-valued sphere decoder to be used would be $2k$. As a result, a precoding scheme with higher worst-case ML-decoding complexity than another precoding scheme will also have higher average ML-decoding complexity. Hence, throughout this paper, we consider only the worst case ML-decoding complexity.

We make use of the following known results, which are needed for our purpose. 
\begin{theorem}\label{lemma1}
 \cite{giannakis} For a scalar channel modelled by $y = \sqrt{SNR}\beta x + n$, where $n \sim \mathcal{N}_{\mathbb{C}}\left(0,1\right)$, $\mathbb{E}[\vert x \vert^2] = 1$ and $\alpha = \vert \beta \vert^2$ is a nonnegative random variable whose probability density function (PDF) $f_\alpha(\alpha)$ is such that 
\begin{equation*}
 f_\alpha(\alpha) = c\alpha^t + o(\alpha^{t}),  ~~~ \textrm{as}~\alpha \to 0^+, 
\end{equation*}
the average symbol error probability (SEP) $P_e$, which is given by  
\begin{equation*}
P_e  =  \mathbb{E}[P_{e,\alpha}] =  \int_{0}^{\infty}Q\left(\sqrt{k\alpha SNR} \right)f_\alpha d\alpha,
\end{equation*}
is such that as $SNR \to \infty$,
\begin{equation*}
 P_e = \frac{2^tc\Gamma(t+\frac{3}{2})}{\sqrt{\pi}(t+1)}(k.SNR)^{-(t+1)} + o\left(SNR^{-(t+1)}\right),
\end{equation*}

\noindent where $k$ is a fixed positive constant depending on the constellation, $c$ is another constant defining the marginal PDF of $\alpha$ and $P_{e,\alpha} = Q\left(\sqrt{k\alpha SNR} \right) $ is the $\alpha$ dependent instantaneous SEP. If $\mathbb{E}[\alpha] = 1$, then, $SNR$ is the average SNR at the receiver and the diversity gain $G_d$ and the coding gain $G_c$ can be defined as
\begin{eqnarray*}
G_d = t+1, ~~~
G_c  =  k\left(\frac{2^tc\Gamma(t+\frac{3}{2})}{\sqrt{\pi}(t+1)}\right)^{-\frac{1}{t+1}}.
\end{eqnarray*}
\end{theorem}

Given that $\sigma_i$, $i=1,2,\cdots,n_{min}$, are the non-zero singular values of $\textbf{H}$, it is known that $\sigma_i^2$ are the non-zero eigenvalues of $\textbf{HH}^H$, which are denoted in the descending order by $\lambda_i$, $i = 1, 2, \cdots, n_{min}$. The following theorem gives the expression for the first order expansion of the marginal PDF of $\lambda_i$ as $\lambda_i \to 0^+$. 
\begin{theorem}\label{lemma2}
\cite{ordonez} Let the entries of the $n_r \times n_t$ matrix $\textbf{H}$ be i.i.d. complex Gaussian with zero mean and unit variance. The first order expansion of the marginal PDF of the $k^{th}$ largest eigenvalue $\lambda_k$ of the complex central Wishart matrix $\textbf{HH}^H$ is given by $f_{\lambda_k}(\lambda_k) = a_k\lambda_k^{d_k} + o\left(\lambda_k^{d_k}\right)$, as $\lambda_k \to 0^+$, $k=1,2,\cdots,n_{min}$, with $d_k = (n_t - k + 1)(n_r-k+1)-1$ and $a_k$ being positive constants.
\end{theorem}

In \eqref{revised_channel_model}, if $\textbf{P} = \textbf{I}_{n_t}$ or $\textbf{P} = [ \textbf{I}_{n_r} ~~ \textbf{O}_{n_r \times n_t-n_r}]^T$, depending on whether $n_{min} = n_t$ or $n_{min} = n_r$, respectively, each of the symbols $x_i$, $i=1,2,\cdots,n_{min}$ will experience a diversity gain given by $G_{d_i} = (n_t - i + 1)(n_r-i+1)$. This is evident from Theorem \ref{lemma1} and Theorem \ref{lemma2}. The above operation of premultiplying the symbol vector by $\textbf{VP}$, with $\textbf{P} = \textbf{I}_{n_t}$ (for $n_t < n_r$) or $\textbf{P} = [ \textbf{I}_{n_r} ~~ \textbf{O}_{n_r \times n_t-n_r}]^T$ (for $n_t \geq n_r$) can be viewed to result in $n_{min}$ {\it virtual subchannels}. So, the overall diversity gain for the symbol vector is $\min\{ G_{d_i}, i=1,2,\cdots,n_{min} \} = (n_{max} - n_{min} +1)$, where $n_{max} = \max(n_t,n_r)$. This is the least diversity order one can obtain in a precoded MIMO system with ML-decoding. However, assuming that the symbols take values from an arbitrary signal constellation of size $M$, the ML-decoding complexity is $\mathcal{O}(M)$, since each symbol can be decoded independently from the others.

Let $\Delta \textbf{x} \triangleq \textbf{x}-\textbf{x}^\prime$, where $\textbf{x},\textbf{x}^\prime \in \mathcal{A}^{n_{min} \times 1}$. 
\begin{theorem}\label{lemma3}\cite{multiple_beamforming}
For $\textbf{P}$ such that $[\textbf{P}\Delta \textbf{x} ]_1 \neq 0$ for any non-zero value of $\Delta \textbf{x} \in \{\textbf{x-x}^\prime | \textbf{x,x}^\prime \in \mathcal{A}^{n_{min}\times 1} \}$, the diversity gain of the system is $n_tn_r$.
\end{theorem}

\begin{proof} The instantaneous probability that a transmitted symbol vector $\textbf{x}$ is falsely decoded to some other vector $\textbf{x}^\prime$ is given by
\begin{equation}\label{false_dec}
Pr\{\textbf{x} \to \textbf{x}^\prime\} = Q\left(\sqrt{\frac{SNR}{2n_tE}} \Vert \textbf{D}\textbf{P}(\textbf{x-x}^\prime)\Vert\right).
\end{equation}
Let $\epsilon_{min} \triangleq \min_{\Delta \textbf{x}} \left\{ \left\vert [\textbf{P}\Delta \textbf{x}]_1 \right\vert \right\}$, with $\Delta \textbf{x} \neq \textbf{O}_{n_{min}\times 1}$. So, the probability $P_e(\textbf{x})$ that a transmitted vector $\textbf{x}$ is falsely decoded is upper bounded as 
\begin{eqnarray}\label{upper_bound}
 P_e(\textbf{x}) & \leq & \left(\vert \mathcal{A} \vert ^{n_{min}}-1\right)Q\left(\sqrt{\frac{SNR}{2n_tE}} \sigma_1\epsilon_{min}\right),
\end{eqnarray}
where $\textbf{D}(1,1) = \sigma_1$, the largest singular value of $\textbf{H}$. Assuming that all the symbol vectors taking values from $\mathcal{A}^{n_{min} \times 1}$ are equally likely to be transmitted, the average instantaneous {\it word error probability} (WEP), dependent on $\textbf{D}$ is given by
\begin{equation}\label{in_wep}
 P_{e,\textbf{D}} =  \frac{1}{\vert \mathcal{A} \vert^{n_{min}}}\sum_{\textbf{x} \in \mathcal{A} ^{n_{min}\times 1}}P_e(\textbf{x}).
\end{equation}
Using \eqref{upper_bound} in \eqref{in_wep}, 
\begin{eqnarray*}
 P_{e,\textbf{D}}  \leq  \left(\vert \mathcal{A} \vert ^{n_{min}}-1\right)Q\left(\sqrt{\frac{SNR}{2n_tE}}  \sigma_1 \epsilon_{min}\right) = \left(\vert \mathcal{A} \vert ^{n_{min}}-1\right)Q\left(\sqrt{ \left(\frac{\epsilon_{min}^2}{2n_tE}\right)\lambda_1SNR}  \right),
\end{eqnarray*}
where $\lambda_1 = \sigma_1^2$. So, from Theorem \ref{lemma1} and Theorem \ref{lemma2}, the average WEP $P_e$ as $SNR \to \infty$ is given by
\begin{equation}\label{div_gain_lattice}
 P_e \leq C.SNR^{-n_tn_r} + o\left(SNR^{-n_tn_r}\right),
\end{equation}
where 
\begin{equation*}
C =  \left(\vert \mathcal{A} \vert ^{n_{min}}-1\right)\frac{a_1(2n_tn_r-1)(2n_tn_r-3)\cdots1}{2n_tn_r}\left( \frac{\epsilon_{min}^2}{2n_tE}\right)^{-n_tn_r},
\end{equation*}
\noindent with $a_1$ being a positive constant such that $f_{\lambda_1}(\lambda_1) = a_1\lambda_1^{n_tn_r-1} + o\left(\lambda_1^{n_tn_r-1}\right)$ as $\lambda_1 \to 0^+$. Note that in obtaining $C$, we have used the fact that $\Gamma(t+1) = t\Gamma(t)$ and $\Gamma(1/2) = \sqrt{\pi}$. Since $\epsilon_{min} > 0$, $C < \infty$ and from \eqref{div_gain_lattice}, the diversity gain achieved by the system is $n_tn_r$.
\end{proof}
An alternative proof of Theorem \ref{lemma3} has been presented in \cite{multiple_beamforming}. Since the steps of our proof are used in Section \ref{sec_div_gain} of this paper, and also for the sake of completeness, we have provided our version of the proof.
 
{\it Note}: The condition that $[\textbf{P}\Delta \textbf{x} ]_1 \neq 0$ for any non-zero value of $\Delta \textbf{x} \in \{\textbf{x-x}^\prime | \textbf{x,x}^\prime \in \mathcal{A}^{n_{min}\times 1} \}$ is only {\it sufficient} to guarantee full-diversity. There might be several precoders which do not satisfy this condition but still give full-diversity. This will be elaborated in Section \ref{sec_div_gain}. Also note that in Theorem \ref{lemma3}, the constraint is only on the first entry of $\textbf{P}\Delta \textbf{x}$. The other entries are allowed to be zeros.

 Obtaining $\textbf{P}$ such that $\epsilon_{min} \neq 0$ is not difficult. Choosing $\textbf{P}$ to be $[\textbf{G} ~ \textbf{O}_{n_r\times (n_t-n_r)} ]^T$  (for $n_t >n_r$) or $\textbf{G}$ (for $n_t \leq n_r$) for QAM constellations, where $\textbf{G} \in \mathbb{R}^{n_{min}\times n_{min}}$ is the rotated $\mathbb{Z}^{n_{min}}$ lattice generator matrix with a non-zero product distance, as presented in \cite{full_div}, ensures that the diversity gain is $n_tn_r$. If $\mathcal{A}$ is a square QAM constellation of size $M$, the ML-decoding complexity is $\mathcal{O}\left(M^{\frac{n_{min}}{2}}\right)$, since all the $n_{min}$ independent symbols are entangled in the decoding metric, but the real part of the symbol vector can be independently decoded from the imaginary part. This is possible because $\textbf{G}$ is real-valued. In \cite{multiple_beamforming}, complex-valued precoders are used to achieve full-diversity and they offer an ML-decoding complexity of $\mathcal{O}\left(M^{n_{min}}\right)$.

\section{Review of Low ML-decoding complexity Precoders}\label{sec_low}
This section gives a brief overview of existing low-complexity precoders. The first precoder is called the E-$d_{min}$ precoder \cite{edmin_precoder}, which is an extension of the MIMO precoder for $n_t = 2$ \cite{optimal_precoder}, developed for $4$-QAM. 

\subsection{E-$d_{min}$ precoder}
The precoder $\textbf{P}$ of size $n_{min} \times n_{min}$ (for $n_t > n_r$, the remaining $n_t-n_r$ rows of $\textbf{P}$ are zeros) has the following structure 

{\footnotesize
\begin{equation}\label{edmin_structure}
 \textbf{P} = \left[\begin{array}{cccccccc}
\textbf{M}_1(1,1) & & & & & & & \textbf{M}_1(1,2) \\
& \textbf{M}_2(1,1) & & & & & \textbf{M}_2(1,2) & \\
& & \ddots & & & \iddots & &  \\
& & & \textbf{M}_{\frac{n_{min}}{2}}(1,1) & \textbf{M}_{\frac{n_{min}}{2}}(1,2) & & & \\
& & & \textbf{M}_{\frac{n_{min}}{2}}(2,1) & \textbf{M}_{\frac{n_{min}}{2}}(2,2) & & & \\
& & \iddots & & & \ddots & &  \\
& \textbf{M}_2(2,1) & & & & & \textbf{M}_2(2,2) & \\
\textbf{M}_1(2,1) & & & & & & & \textbf{M}_1(2,2) \\
\end{array}\right],
\end{equation}
}
where, if $\gamma_i \triangleq \tan^{-1}\left( \frac{\sigma_{n_{min}-i+1}}{\sigma_i}\right)$ is such that $0 < \gamma_i < \gamma_o $, then,
\begin{equation}\label{precoder_single}
 \textbf{M}_i =  \sqrt{\frac{2n_t \tau_i^2}{n_{min}}}\left[ \begin{array}{cc}
                         \sqrt{\frac{3+\sqrt{3}}{6}} &  \sqrt{\frac{3-\sqrt{3}}{6}}e^{j\pi/12} \\
              0 & 0\\
                         \end{array}\right]
\end{equation}
and if $ \gamma_o \leq \gamma_i \leq \pi/4$,
\begin{equation*}
 \textbf{M}_i =  \sqrt{\frac{n_t \tau_i^2}{n_{min}}}\left[ \begin{array}{cc}
                         \cos \psi_i &  0 \\
              0 & \sin \psi_i \\
                         \end{array}\right]\left[ \begin{array}{rr}
                         1 &  e^{j\pi/4} \\
              -1 & e^{j\pi/4}\\
                         \end{array}\right],
\end{equation*}
where, 
\begin{equation*}
\psi_i   =  \tan^{-1}\left(\frac{\sqrt{2}-1}{\cos \gamma_i}\right), ~~~
 \gamma_0 =  \textrm{ tan}^{-1}\left(\sqrt{\frac{3\sqrt{3}-2\sqrt{6}+2\sqrt{2}-3}{3\sqrt{3}-2\sqrt{6}+1}} \right) \approx 0.3016
\end{equation*}
and $\tau_i = \sqrt{\frac{n_{min}}{2}\left(\rho_i^2\delta(\gamma_i)\sum_{j=1}^{n_{min}}\frac{1}{\rho_j^2\delta(\gamma_j)}\right)^{-1}}$, with $\rho_i = \sqrt{\sigma_i^2+\sigma_{n_{min}-i+1}^2}$ and 
\begin{equation*}
 \delta(\gamma_j) =  \left\{\begin{array}{ll}
                  (1-\frac{1}{\sqrt{3}})\cos^2 \gamma_j, & \textrm{if}~ 0 < \gamma_j < \gamma_o\\
                   \frac{(4-2\sqrt{2})\cos^2\gamma_j \sin^2 \gamma_j}{1 +(2-2\sqrt{2})\cos^2 \gamma_j}, & \textrm{otherwise}.\\
                   \end{array}\right.
\end{equation*}
The precoder essentially entangles the virtual subchannels with index $i$ and $n_{min}-i+1$, $i = 1,2,$ $\cdots,n_{min}/2$\footnote[3]{for odd valued $n_{min}$, $\frac{n_{min}}{2}$ is replaced by $\lfloor\frac{n_{min}}{2}\rfloor$ and the $\left(\lfloor\frac{n_{min}}{2}\rfloor + 1\right)^{th}$ subchannel is left unpaired.}. Such a scheme will have an ML-decoding complexity of $\mathcal{O}\left( M\sqrt{M}\right)$. It has been shown that the scheme guarantees a diversity gain equal to $(n_t-\frac{n_{min}}{2}+1)(n_r-\frac{n_{min}}{2}+1)$. Also, the precoder is optimal among precoders based on the SVD of the channel for $n_{min} =2$ and $4$-QAM \cite{optimal_precoder}. 

\subsection{X-precoder}
The $X$-precoder has the same structure as in \eqref{edmin_structure}, with the matrices $\textbf{M}_i$ given as 
\begin{equation*}
 \textbf{M}_i = \sqrt{\frac{n_t}{n_{min}}}\left[\begin{array}{rr}
                        \cos\theta_i & -\sin\theta_i\\
\sin \theta_i & \cos\theta_i \\
                       \end{array}\right],
\end{equation*}
where, for $4$-QAM,
\begin{equation*}
 \theta_i = \left\{\begin{array}{ll}
                 \pi/4,  & \textrm{if} ~ \gamma_i \geq \pi/3 \\
                 \tan^{-1}\left(\frac{1-\tan^2\gamma_i - \sqrt{1+\tan^4\gamma_i-3\tan^2\gamma_i}}{\tan^2\gamma_i}\right)  & \textrm{otherwise}.\\
                   \end{array}\right.
\end{equation*}
This scheme has also been shown to guarantee a diversity gain equal to 
$(n_t-\frac{n_{min}}{2}+1)(n_r-\frac{n_{min}}{2}+1)$, but has an ML-decoding complexity of $\mathcal{O}\left( \sqrt{M}\right)$ only (refer Subsection \ref{subsec_compexity} for details). However, it is expected to lose out in performance for $4$-QAM when compared with the E-$d_{min}$ precoder, since it is not optimal. Also, an explicit expression for the precoder when $M >4$ does not exist.

\subsection{Y-precoder}
The $Y$-precoder \cite{xy_codes} has the $Y$-structure but it uses a displacement vector and its precoded symbol vector $\textbf{s}$ can be written as 
\begin{equation}\label{ystructure}
 {\bf s} = {\bf V(Px + u)},
\end{equation}
where, $\textbf{u}$ is the displacement vector. The precoded vector can also be expressed as 
\begin{equation*}
 {\bf s}  = \textbf{VP}_{eff}\textbf{x}_{eff} ,
\end{equation*}
where, $\textbf{Px}+\textbf{u} = \textbf{P}_{eff}\textbf{x}_{eff}$, with $\textbf{P}_{eff}$ and $\textbf{x}_{eff}$ being the effective precoder and the effective symbol vector, respectively. These are defined as 
\begin{equation*}
 \textbf{P}_{eff} = \operatorname{diag}\left(a_1,a_2,\cdots,a_{\frac{n_{min}}{2}},b_{\frac{n_{min}}{2}},\cdots,b_2,b_1\right)
\end{equation*}
where,
\begin{equation*}
 (a_i,b_i) = \left\{\begin{array}{ll}
                 \left(\sqrt{\frac{3n_t}{n_r\left(M^2-1\right)}},0 \right),  & \textrm{if} ~ \beta_i^2 \geq \frac{M^2-1}{3} \\
                 \left(\sqrt{\frac{n_t}{3n_r\left( \beta_i^2+M^\prime\right)}}, \beta_i\sqrt{\frac{n_t}{n_r\left( \beta_i^2 + M^\prime \right)}} \right)  & \textrm{otherwise},\\
                   \end{array}\right.
\end{equation*}
and $ M^\prime = \frac{M^2-1}{9}$ and $\beta_i = \frac{\sigma_i}{\sigma_{n_{min}-i+1}}$. The constellation $\mathcal{A} \in \mathbb{Z}^{2\times1}$ of size $M$ is two-dimensional with 
the signal vectors (not to be confused with the symbol vector ${\bf x}_{eff}$) $\textbf{z}_l$, $l=1,2,\cdots, M$ defined as
\begin{equation*}
\textbf{z}_l  =  \left[\begin{array}{c}
 2l-M-1 \\
(-1)^l\\
\end{array}\right]
\end{equation*}
and the symbol vector that is associated with the $(i,n_{min}-i+1)$ subchannel pairing is $ \textbf{s}_i \triangleq \textbf{v}_i + j \textbf{v}_{n_{min}-i+1}$, with  $\textbf{v}_i, \textbf{v}_{n_{min}-i+1} \in \mathcal{A}$, $i= 1,2,\cdots,n_{min}/2 $. Hence,
\begin{equation*}
 \textbf{x}_{eff} = \left[ [\textbf{s}_1]_1, [\textbf{s}_2]_1, \cdots, [\textbf{s}_{\frac{n_{min}}{2}}]_1, [\textbf{s}_{\frac{n_{min}}{2}}]_2, \cdots, [\textbf{s}_2]_2,[\textbf{s}_1]_2  \right]^T.
\end{equation*}
So, the effective precoder of the $Y$-precoder is a diagonal matrix, while $\textbf{P}$, as given in \eqref{ystructure}, has the '$Y$' structure. The $Y$-precoder has been shown to have better error performance than the $X$-precoder for ``ill-conditioned'' channels, i.e., for low values of $\frac{\sigma_{n_{min}-i +1}}{\sigma_i}$, $i=1,2,\cdots,n_{min}/2$, while for well-conditioned channels, the $X$-precoder has better error performance. However, the $Y$-precoder has lower ML-decoding complexity, which is $\mathcal{O}(1)$. Hence, among all existing precoders, the $Y$-precoder has the least ML-decoding complexity while the E-$d_{min}$ precoder has the best performance for $4$-QAM.

\section{SVD-based, Approximately Optimal, Real-valued Precoder for $n_t =2$}\label{sec_proposed}
In this section, we propose a real-valued precoder for 2 transmit antennas and QAM constellations. The precoder is approximately optimal among the SVD based real-valued precoders for QAM constellations. The primary advantage of this precoder over the complex-valued optimal precoder \cite{optimal_precoder} is that it is much easier to find the entries of the precoder for larger constellations, since it has only 2 parameters that need to be searched for, while the complex-valued precoder has 3 parameters. Without loss of generality, we consider 2 receive antennas and 2 transmit antennas, for which $\textbf{D}$ in \eqref{revised_channel_model} can be expressed as 
\begin{equation*}
\textbf{D} =  \rho\left[\begin{array}{cc}
  \cos \gamma & 0\\
  0    &  \sin \gamma \\
\end{array}\right],
\end{equation*}
where $ \rho = \sqrt{\sigma_1^2+\sigma_2^2}$ and $ \gamma = \tan^{-1} (\frac{\sigma_2}{\sigma_1})$. Clearly, $0 < \gamma \leq \pi/4$. Let
\begin{equation}\label{e_min}
 E_{min}(\textbf{P}) \triangleq \min_{\Delta \textbf{x} \neq \textbf{O}_{2\times1}} \left\{ \Vert \textbf{DP}\Delta \textbf{x} \Vert^2 ,  ~~\Delta \textbf{x} \in \{\textbf{x} -\textbf{x}^\prime ~|~ \textbf{x,~x}^\prime \in \mathcal{A}^{2\times1} \right\}.
\end{equation}
From \eqref{false_dec}, the optimal precoder is given by $\textbf{P}^{opt} = \underset{\textbf{P}}{\operatorname{argmax}}\{E_{min}(\textbf{P})\}$, which may or may not be unique. In \cite{optimal_precoder}, $\textbf{P}^{opt} \in \mathbb{C}^{2\times2}$ was obtained for $4$-QAM as follows. Using SVD, $\textbf{P} \in \mathbb{C}^{2\times2}$ can be written as $\textbf{P} = \textbf{A}\bf{\Sigma}\textbf{B}^H$, where $\textbf{A}$ is a unitary matrix of size $2\times2$ and 
\begin{equation}\label{opt_complex}
{\bf \Sigma} = \sqrt{2}\left[\begin{array}{cc}
 \cos \psi & 0\\
  0    &  \sin \psi \\
\end{array}\right], ~~ \textbf{B}^H = \left[\begin{array}{rr}
 \cos\theta & -\sin\theta \\
\sin\theta & \cos\theta \\
\end{array}\right] \left[\begin{array}{cc}
 1& 0 \\
0 & e^{j\phi} \\
\end{array}\right].
\end{equation}
For QAM constellations, because of the symmetry associated with the constellation, $0 \leq \theta \leq \pi/4$, $0 \leq \psi \leq \pi/2$ and $0 \leq \phi \leq \pi/2$. It was shown in \cite{optimal_precoder} that $\textbf{A}$ can be taken to be identity without affecting the optimality. Using numerical search, the optimal values for $\theta$, $\psi$ and $\phi$ were found out for $4$-QAM. However, there are two major obstacles when this method is used for larger QAM constellations. Firstly, numerical search becomes practically hard for larger constellations due to the fact that there are three parameters to be searched for. Secondly, numerical searches do not give a closed form expression for the optimal angles and the method employed in \cite{optimal_precoder} to obtain closed form expressions for the optimal angles for $4$-QAM is not amenable for application to larger QAM constellations. Due to these limitations, we look for a real-valued optimal precoder which also naturally offers lower ML-decoding complexity (this is elaborated in Subsection \ref{subsec_compexity}). A real-valued precoder can be expressed as $\textbf{P}(\psi,\theta) = {\bf A\Sigma B}^T$ where ${\bf A}$ can be taken to be identity without affecting optimality and 
\begin{equation*}
{\bf \Sigma} = \sqrt{2}\left[\begin{array}{cc}
 \cos \psi & 0\\
  0    &  \sin \psi \\
\end{array}\right], ~~ \textbf{B} = \left[\begin{array}{rr}
 \cos\theta & \sin\theta \\
-\sin\theta & \cos\theta \\
\end{array}\right].
\end{equation*}
Note that there are only two parameters to be searched for. Our approach towards finding the optimal precoders is also based on numerical search, but the method to obtain closed form expressions for the optimal angles is novel and easily applicable for any $M$-QAM. However, since this method is based on numerical search, it is not known if the angles are exactly optimal. Finding the exactly optimal values of $\theta$ and $\psi$ as a function of $\gamma$ involves an {\it exhaustive} search over the range of $\theta$ and $\psi$, which is practically impossible. However, a numerical search, with $\theta$ and $\psi$ varying in very small increments, gives the values of $\theta$ and $\psi$, which we denote by $\theta^*$ and $\psi^*$, respectively, such that $E_{min}(\textbf{P}(\psi^*,\theta^*))$ is nearly equal to $E_{min}(\textbf{P}^{opt})$, with $\textbf{P}^{opt}$ being the optimal real-valued precoder. For this reason, we call our precoder approximately optimal.

A square QAM signal set (not necessarily Gray coded) of size $M$ is given by 
\begin{equation}\label{qam_defn}
\mathcal{A}_{M-QAM} = \{ a + jb ~|~ a,b \in \mathcal{A}_{\sqrt{M}-PAM} \}, 
\end{equation}
where $\mathcal{A}_{\sqrt{M}-PAM} = \{ 2i-\sqrt{M}-1, i=1,2,\cdots,\sqrt{M} \}$ is a PAM constellation of size $\sqrt{M}$. Let 

\begin{equation*}
 \textbf{F}(\gamma,\psi,\theta) \triangleq \left[\begin{array}{cc}
     \cos\gamma & 0\\
0 & \sin\gamma \\                         
\end{array}\right]\left[\begin{array}{cc}
     \cos\psi & 0\\
0 & \sin\psi \\                         
\end{array}\right]\left[\begin{array}{rr}
     \cos\theta & -\sin\theta\\
\sin\theta & \cos\theta \\                         
\end{array}\right],
\end{equation*}
and 
\begin{equation}\label{delta_min}
 \delta(\gamma,\mathcal{A}) = \max_{(\psi,\theta)} \left\{ \min_{\Delta \textbf{x}| \Delta \textbf{x} \neq \textbf{O}_{2\times1}} \left\{ \Vert \textbf{F}(\gamma,\psi,\theta) \Delta \textbf{x}\Vert^2 , ~~ ~
 \Delta \textbf{x} \in \left\{\textbf{x,x}^\prime ~|~\textbf{x,x}^\prime \in \mathcal{A}^{2\times1} \right \}\right\} \right\},
\end{equation}
where, for our numerical search, we take $\psi = \varDelta.k,~k=1,2,\cdots,\left\lfloor\frac{\pi}{2\varDelta}\right\rfloor$, $
\theta= \varDelta.k,~k=1,2,\cdots,\left\lfloor\frac{\pi}{4\varDelta}\right\rfloor$, with $\varDelta$ being the increment size, taken to be 0.001 radians for our searches. Let
\begin{equation}\label{opt_theta_psi}
 (\psi^*,\theta^*) = \underset{(\psi,\theta)}{\operatorname{argmax}}\left\{ \min_{\Delta \textbf{x}, \Delta \textbf{x} \neq \textbf{O}_{2\times1}} \left\{ \Vert \textbf{F}(\gamma,\psi,\theta) \Delta \textbf{x}\Vert^2 \right\} \right\}.
\end{equation}

 We note that for $M$-QAM, $E_{min}(\textbf{P}(\psi^*,\theta^*)) = 2\rho^2\delta\left(\gamma,\mathcal{A}_{M-QAM}\right) = 2\rho^2\delta\left(\gamma,\mathcal{A}_{\sqrt{M}-PAM}\right)$. Hence, we only need to search for $\theta^*$ and $\psi^*$ for which $\delta(\gamma,\mathcal{A}_{\sqrt{M}-PAM})$ is obtained. Note that this simplification of the search to only a $\sqrt{M}$-PAM is possible since $\textbf{F}(\gamma,\psi,\theta)$ is real-valued. This is another huge advantage over the complex-valued precoder, which does not enjoy this benefit. Henceforth, $\theta^*$ and $\psi^*$ are used to denote the approximately optimal angles of $\theta$ and $\psi$. Due to our choice of the increment size, one can safely say that $\left( E_{min}(\textbf{P}^{opt})-E_{min}(\textbf{P}(\theta^*,\psi^*)) \right) < \kappa.E_{min}(\textbf{P}^{opt})$, where $\kappa$ is a very small fraction of the order of $10^{-3}$. 

The search results reveal that $\theta^*$ as a function of $\gamma$ can be written as 
\begin{equation}\label{theta_function}
\theta^* =  \sum_{k=1}^n \theta_k^* \left( u(\gamma - \gamma_{k}^\prime) -u(\gamma - \gamma_{k}^\prime-w_k)  \right),
\end{equation}
where $\theta_k^*, k=1,\cdots,n$, are constants, $n$ is the finite number of different values $\theta^*$ takes, $\gamma_k^\prime$ is the value of $\gamma$ at which $\theta^*$ changes from $\theta_{k-1}^*$ to $\theta_k^*$, with $\gamma_1^\prime = 0$, $\theta_0^* = 0$, $w_k = \gamma_{k+1}^\prime -  \gamma_k^\prime$ and $\gamma_{n+1}^\prime = \pi/4$. The search results also reveal that $\psi^*$ cannot be expressed as a weighted sum of shifted step functions and hence a closed form expression needs to be obtained analytically. To obtain this, we first obtain $\theta^*$ as follows.

\subsection{Calculating $\theta^*$}

For $M$-QAM, in order to obtain $\theta^*$ and $\psi^*$, as given by \eqref{opt_theta_psi}, the entries of $\Delta \textbf{x}$ take values from $\left\{ 2\left(-\sqrt{M}+i\right), i = 1,2,\cdots,2\sqrt{M}-1\right\}$. Let $p,q\in \{ -\sqrt{M}+i, i = 1,2,\cdots,2\sqrt{M}-1\}$ be such that

\begin{equation}\label{eq_delmin}
 4\Vert\textbf{F}\left(\gamma,\psi^*,\theta^*\right)[p~q]^T \Vert^2 = \delta(\gamma, \mathcal{A}_{M-QAM}) = \delta(\gamma, \mathcal{A}_{\sqrt{M}-PAM}).
\end{equation}

 The numerical searches done for 5 QAM constellations - $4$-/$16$-/$64$-/$256$-/$1024$-QAM reveal that 
\begin{enumerate}
 \item there are two distinct $(p,q)$ pairs for which \eqref{eq_delmin} is satisfied when $0 < \gamma \leq \gamma_2^\prime$, where $\gamma_2^\prime$ is as defined in \eqref{theta_function}. These are $(0,1)$ and $(1,\sqrt{M}-1)$. Also $\psi^* = 0$ in this range of $\gamma$.
\item There are three distinct $(p,q)$ pairs for which \eqref{eq_delmin} is satisfied when $\gamma_k^\prime \leq \gamma \leq \gamma_{k+1}^\prime$, $k = 2,\cdots,n$.
\end{enumerate}
Let \begin{equation*}
 \varepsilon(p,q,\theta^*,\psi^*) \triangleq \cos^2\gamma\cos^2\left(\psi^*\right)(p\cos\left(\theta^*\right) -q\sin\left(\theta^*\right))^2 + \sin^2\gamma\sin^2\left(\psi^*\right)(q\cos\left(\theta^*\right)+ p\sin\left(\theta^*\right))^2.
\end{equation*}
So, for $0 < \gamma \leq \gamma_2^\prime$, we have
\begin{equation*}
  \varepsilon(0,1,\theta_1^*,0) = \varepsilon(1,\sqrt{M}-1,\theta_1^*,0),
\end{equation*}
solving which we obtain $\theta_1^* = \tan^{-1}\frac{1}{\sqrt{M}}$. The other solution, which is  $\theta_1^* = \tan^{-1}\left(\frac{1}{\sqrt{M}-2}\right)$, is ruled out since it has been observed that $E_{min}\left(\textbf{P}\left(0,\tan^{-1}\left(1/\sqrt{M}\right)\right)\right) > E_{min}\left(\textbf{P}\left(0,\tan^{-1}\left(1/(\sqrt{M}-2)\right)\right)\right)$ for $0 < \gamma \leq \gamma_2^\prime$. For $\gamma_k^\prime \leq \gamma \leq \gamma_{k+1}^\prime$, $k = 2,\cdots,n$, we have
\begin{equation*}
 \varepsilon(p_1,q_1,\theta_k^*,\psi^*) = \varepsilon(p_2,q_2,\theta_k^*,\psi^*) = \varepsilon(p_3,q_3,\theta_k^*,\psi^*),
\end{equation*}
where $(p_1,q_1)$, $(p_2,q_2)$ and $(p_3,q_3)$ are the three pairs for which \eqref{eq_delmin} is satisfied. Solving them, we arrive at 
\begin{equation}\label{eq1}
 \tan^2\gamma\tan^2\left(\psi^*\right) = 1 + \frac{p_1^2+q_1^2-p_2^2-q_2^2}{(p_2q_2-p_1q_1)\sin\left(2\theta_k^*\right) + \left( q_2^2-q_1^2 \right)\cos^2\left(\theta_k^*\right) + \left(p_2^2-p_1^2\right) \sin^2\left(\theta_k^*\right)},
\end{equation}
\begin{equation}\label{eq2}
 \tan^2\gamma\tan^2\left(\psi^*\right) = 1 + \frac{p_1^2+q_1^2-p_3^2-q_3^2}{(p_3q_3-p_1q_1)\sin\left(2\theta_k^*\right) + \left( q_3^2-q_1^2 \right)\cos^2\left(\theta_k^*\right) + \left(p_3^2-p_1^2\right) \sin^2\left(\theta_k^*\right)}.
\end{equation}
Equating \eqref{eq1} and \eqref{eq2}, we obtain
\begin{equation}\label{theta_eq}
 (a_1d_2-a_2d_1)\tan^2\left(\theta_k^*\right) + 2(a_1b_2-a_2b_1)\tan\left(\theta_k^*\right) + a_1c_2-a_2c_1 = 0,
\end{equation}
where $a_1 = p_1^2+q_1^2-p_2^2-q_2^2$, $b_1 = p_2q_2-p_1q_1$, $c_1 = q_2^2-q_1^2$, $d_1 = p_2^2-p_1^2$, $a_2 = p_1^2+q_1^2-p_3^2-q_3^2$, $b_2 = p_3q_3-p_1q_1$, $c_2 = q_3^2-q_1^2$ and $d_2 = p_3^2-p_1^2$. Equation \eqref{theta_eq} has been observed to have only one solution in the range $(0,\pi/4)$. This solution gives $\theta_k^*$.

\subsection{Calculating $\psi^*$} \label{calculate_psi}

As mentioned before, $\psi^* = 0$ for $0 < \gamma \leq \gamma_2^\prime$. In order to obtain $\psi^*$ for $\gamma_k^\prime \leq \gamma \leq \gamma_{k+1}^\prime$, $k=2,\cdots,n$, we note from \eqref{eq1} and \eqref{eq2} that $\tan^2\gamma\tan^2\left(\psi^*\right)$ is constant in that range of $\gamma$ and hence,

\begin{equation}\label{tan_gammatan_phi}
 \psi^* = \tan^{-1}\left( \frac{\sqrt{A_k}}{\tan\gamma}\right),
\end{equation}
where $A_k$ is given by the R.H.S of \eqref{eq1} (or \eqref{eq2}).

\subsection{Calculating $\gamma_k^\prime$}
Having obtained $\theta^*$ and $\psi^*$, we proceed to find the exact values of $\gamma_k^\prime$, $k=2,\cdots,n$ as follows. For convenience, let $\psi^*(\theta_k^*,\gamma) \triangleq \psi^*$ (given by \eqref{tan_gammatan_phi}) for $\gamma_k^\prime \leq \gamma \leq \gamma_{k+1}^\prime$, $k=1,\cdots,n$. Since $\theta^*$ is discontinuous at $\gamma_k^\prime$, where it makes a transition from $\theta_{k-1}^*$ to $\theta_{k}^*$, $k \geq 2$, we have
\begin{equation*}
 \varepsilon\left(p_{k-1},q_{k-1},\theta_{k-1}^*,\psi^*(\theta_{k-1}^*,\gamma_k^\prime)\right) = \varepsilon\left(p_{k},q_{k},\theta_{k}^*,\psi^*(\theta_{k}^*,\gamma_k^\prime)\right),
\end{equation*}
where the pairs $(p_{k-1},q_{k-1})$ and $(p_{k},q_{k})$ satisfy \eqref{eq_delmin} for $\gamma_{k-1}^\prime \leq \gamma \leq \gamma_{k}^\prime$ and $\gamma_k^\prime \leq \gamma \leq \gamma_{k+1}^\prime$, respectively. So, we have
\begin{equation}\label{eq4}
 \cos^2\left(\psi^*(\theta_{k-1}^*,\gamma_k^\prime)\right) = \cos^2\left(\psi^*(\theta_{k}^*,\gamma_k^\prime)\right)\left(\frac{c+A_{k}d}{a+A_{k-1}b}\right),
\end{equation}
\begin{equation}\label{eq5}
 \sin^2\left(\psi^*(\theta_{k-1}^*,\gamma_k^\prime)\right) = \sin^2\left(\psi^*(\theta_{k}^*,\gamma_k^\prime)\right)\left(\frac{A_{k-1}}{A_{k}}\right)\left(\frac{c+A_{k}d}{a+A_{k-1}b}\right),
\end{equation}
where $a = \left(p_{k-1}\cos\left(\theta_{k-1}^* \right) -q_{k-1}\sin\left(\theta_{k-1}^* \right)\right)^2 $, $b = \left(q_{k-1}\cos\left(\theta_{k-1}^* \right)+ p_{k-1}\sin\left(\theta_{k-1}^* \right)\right)^2$,\\ $c=\left(p_{k}\cos\left(\theta_{k}^* \right)-q_{k}\sin\left(\theta_{k}^* \right)\right)^2$, $d=\left(q_{k}\cos\left(\theta_{k}^* \right) +p_{k}\sin\left(\theta_{k}^* \right)\right)^2$, and as explained in Subsection \ref{calculate_psi}, $A_{k-1}$ and $A_{k}$ are constants given by $A_{k-1} = \tan^2\gamma\tan^2\left( \psi^*(\theta_{k-1}^*,\gamma)\right)$ for $\gamma_{k-1}^\prime \leq \gamma \leq \gamma_k^\prime$, $ A_{k}= \tan^2\gamma\tan^2\left( \psi^*(\theta_{k}^*,\gamma)\right)$ for $\gamma_{k}^\prime \leq \gamma \leq \gamma_{k+1}^\prime$. Solving \eqref{eq4} and \eqref{eq5}, we obtain,
\begin{equation}\label{eq6}
\psi^*(\theta_{k}^*, \gamma_k^\prime) = \sin^{-1}\sqrt{\left(\frac{A_{k-1}}{A_{k}}-1\right)^{-1} \left( \frac{a+A_{k-1}b}{c+A_{k}d} -1\right)}.
\end{equation}
Using \eqref{tan_gammatan_phi} and \eqref{eq6},
\begin{equation*}
 \gamma_k^\prime = \tan^{-1} \left( \frac{\sqrt{A_k}}{\tan\left(\psi^*\left(\theta_{k}^*, \gamma_k^\prime\right)\right)}\right).
\end{equation*}

\noindent The value of $\delta(\gamma,\mathcal{A}_{M-QAM})$ as defined in \eqref{delta_min} for $\gamma_k^\prime \leq \gamma  \leq \gamma_{k+1}^\prime$ is given by
\begin{equation}\label{explicit_delmin}
 \delta(\gamma,\mathcal{A}_{M-QAM}) = \sin^2\gamma\left(\frac{c + A_kd}{A_k + \tan^2\gamma}\right), 	
\end{equation}
where $c = \left(p_k\cos\left(\theta_k^*\right) - q_k\sin\left(\theta_k^*\right)\right)^2$, $d = \left(q_k\cos\left(\theta_k^*\right) + p_k\sin\left(\theta_k^*\right)\right)^2$, with $(p_k,q_k)$ any of the $(p,q)$ pairs satisfying \eqref{eq_delmin}.

Table \ref{tableqam} presents the values of $\theta^*$ for different values of $\gamma$ for $4$-QAM, $16$-QAM, $64$-QAM, $256$-QAM and $1024$-QAM. The value of the constants $\tan\gamma\tan\left(\psi^*\right)$ and the corresponding pairs $(p,q)$ for which \eqref{eq_delmin} is satisfied are also tabulated. Except for the case of $4$-QAM, the values presented in Table \ref{tableqam} are the approximately optimal values rounded off to the fourth decimal. This has been done since it is very cumbersome to express them in the exact form. All angles are expressed in radians. Noting the values of $\theta^*$ for $4$-QAM, it is natural to believe that the angles tabulated are optimal for $4$-QAM. Also, it can be noted that for every subsequent larger constellation, $\theta^*$ differs from its corresponding values for the lower-sized constellation only at low values of $\gamma$, meaning which the numerical search need not be done over the entire range of $\gamma$ as the size of the constellation increases. The plots of $\delta(\gamma,\mathcal{A}_{M-QAM})$ as a function of $\gamma$ for the different unnormalized QAM constellations are given in Fig. \ref{min_dist_plots}. The curves for $256$- and $1024$-QAM appear to coincide, since they differ only at extremely low values of $\gamma$. In Fig. \ref{min_dist4}, the plots\footnote[4]{In all the plots, the E-$d_{min}$ precoder, our precoder and the $X$-precoder use $M$-QAM, while the $Y$-precoder uses a two-dimensional codebook of size $M$, as defined in \cite{xy_codes}.} of $\delta(\gamma,\mathcal{A})$ for the E-$d_{min}$ precoder, the proposed precoder, the $X$-precoder and the $Y$-precoder are given for $M=4$ with the same power constraint for all the precoders as for our precoder. As was expected, the E-$d_{min}$ precoder has the best values of $\delta(\gamma,\mathcal{A})$ over the entire range of $\gamma$ while our precoder has better values of $\delta(\gamma,\mathcal{A})$ than the $X$- and $Y$-precoders. Fig. \ref{min_dist16} and Fig. \ref{min_dist64} show the plots of $\delta(\gamma,\mathcal{A})$ for our precoder, the $X$-precoder and the $Y$-precoder for $M = 16$ and $M=64$, respectively. For the $X$-precoder, the plots were obtained using numerical searches to obtain the approximately optimal angle for each value of $\gamma$ in the range $(0,\pi/4)$, with $\gamma$ increasing in step sizes of 0.001. Note that for low values of $\gamma$, our precoder and the $Y$-precoder have identical $\delta(\gamma,\mathcal{A})$, which is because both transmission schemes are effectively the same in this range of $\gamma$. With an increase in the constellation size, the $Y$-precoder has increasingly lower values of $\delta(\gamma,\mathcal{A})$ than that of our precoder and the $X$-precoder at higher values of $\gamma$. It is also clear from the plots that the $Y$-precoder is expected to have better error performance than the $X$-precoder only for ill-conditioned channels, i.e., for low values of $\gamma$.

\subsection{ML-decoding complexity}\label{subsec_compexity}
We make use of the following lemma to analyze the ML-decoding complexity of our precoder.
\begin{lemma}\label{lemma4}
 For symbols $x_1$ and $x_2$ taking values from $\mathcal{A}_{M-QAM}$, the symbol $ax_1 + bx_2$ takes values from $\mathcal{A}_{M^2-QAM}$ if $a = \sqrt{M}$, $b=1$ or $b = \sqrt{M}$, $a=1$.
\end{lemma}
\begin{proof}
 Firstly, $A_{M-QAM}$ represents the standard, unnormalized $M$-QAM constellation, as given in \eqref{qam_defn}. Let $\sqrt{M}\mathcal{A}_{M-QAM}$ denote the $M$-QAM constellation scaled by $\sqrt{M}$. So, the distance between any two adjacent signal points on the same vertical or horizontal line of $\sqrt{M}\mathcal{A}_{M-QAM}$ is $2\sqrt{M}$. Now, the constellation given by 
\begin{equation}\label{eq_constellation}
 \mathcal{A} = \left\{\sqrt{M}x_1 + x_2 ~|~ x_1, x_2 \in \mathcal{A}_{M-QAM} \right\}
\end{equation}
can be viewed to be obtained by replacing every element of $\sqrt{M}\mathcal{A}_{M-QAM}$ by the entire constellation $\mathcal{A}_{M-QAM}$ such that the origin of $\mathcal{A}_{M-QAM}$ is the signal point being replaced. Hence, $\mathcal{A}$ has $M^2$ signal points and a QAM structure, and the distance between adjacent points on the same vertical or horizontal line is 2. Therefore, $\mathcal{A}$ is an $M^2$-QAM. 
\end{proof}

The following theorem gives the ML-decoding complexity of the precoder.
\begin{theorem}\label{thm2}
For the proposed precoder, the following claims hold.
\begin{enumerate}
 \item The ML-decoding complexity is $\mathcal{O}(\sqrt{M})$, when $\gamma_k^\prime \leq \gamma \leq \gamma_{k+1}^\prime$, $k=2,3,\cdots,n$.
\item The ML-decoding complexity is the same as that of a real scalar channel when $0 < \gamma \leq \gamma_{2}^\prime$, with no exhaustive search over all the signal points required.
\end{enumerate}
\end{theorem}
\begin{proof}
These claims are proved below.\\
{\it Case 1} : $\gamma_k^\prime \leq \gamma \leq \gamma_{k+1}^\prime, ~~ k=2,3,\cdots,n$.

 In this case, the decoded signal vector $\check{\textbf{x}}$ is 
\begin{equation}\label{dec_1}
\check{\textbf{x}} =  \underset{\textbf{x} \in \mathcal{A}_{M-QAM}^{2\times1}}{\operatorname{argmin}} \left\{\left \Vert \textbf{y}^\prime - \sqrt{\frac{SNR}{2E_M}}\textbf{DPx} \right\Vert^2 \right\} = \underset{\textbf{x} \in \mathcal{A}_{M-QAM}^{2\times1}}{\operatorname{argmin}} \left\{\left \Vert \textbf{y}^{\prime\prime} - \sqrt{\frac{SNR}{2E_M}}\textbf{Rx} \right\Vert^2 \right\},
\end{equation}
where $\textbf{y}^\prime$, $\textbf{D}$ and $\textbf{P}$ are as defined in \eqref{revised_channel_model}, $E_M = 2(M-1)/3$ is the average energy of an $M$-QAM and $\textbf{y}^{\prime\prime} = \textbf{Q}^T\textbf{y}^\prime$, with $\textbf{Q}$ and $\textbf{R}$ obtained on the QR-decomposition of $\textbf{DP}$. Since $\textbf{D}$ and $\textbf{P}$ are real-valued, \eqref{dec_1} can be written as $\check{\textbf{x}} = \check{\textbf{x}}_I +j\check{\textbf{x}}_Q$, where
\begin{equation*}
\check{\textbf{x}}_I =  \underset{\textbf{x}_I \in \mathcal{A}_{\sqrt{M}-PAM}^{2\times1}}{\operatorname{argmin}} \left\{\left \Vert \textbf{y}_{I}^{\prime\prime} - \sqrt{\frac{SNR}{2E_M}}\textbf{Rx}_{I} \right\Vert^2 \right\}, ~~ \check{\textbf{x}}_Q = \underset{\textbf{x}_Q \in \mathcal{A}_{\sqrt{M}-PAM}^{2\times1}}{\operatorname{argmin}} \left\{\left \Vert \textbf{y}_{Q}^{\prime\prime} - \sqrt{\frac{SNR}{2E_M}}\textbf{Rx}_{Q} \right\Vert^2 \right\}, 
\end{equation*}
with 
\begin{equation*}
\textbf{y}^{\prime\prime} = \textbf{y}_{I}^{\prime\prime} + j\textbf{y}_{Q}^{\prime\prime} = \left[y_{1I}^{\prime\prime}+jy_{1Q}^{\prime\prime} ,~ y_{2I}^{\prime\prime}+jy_{2Q}^{\prime\prime} \right]^T, ~~ \check{\textbf{x}} \triangleq \check{\textbf{x}}_{I} + j\check{\textbf{x}}_{Q} = \left[\check{x}_{1I}+j\check{x}_{1Q}, ~ \check{x}_{2I}+j\check{x}_{2Q}\right]^T. 
\end{equation*}

To obtain $\check{\textbf{x}}_I$, instead of using a 2-dimensional real sphere decoder, we do the following. For each possible value of $x_{2I} \in \mathcal{A}_{\sqrt{M}-PAM}$, the corresponding value of $x_{1I}$ is evaluated as 

\begin{equation}\label{quantize}
 x_{1I} = \min\left(\max\left(2.\operatorname{rnd}\left[\frac{u+1}{2} \right]-1,-\sqrt{M}+1\right),\sqrt{M}-1 \right),
\end{equation}
where
\begin{equation*}
u = \frac{\sqrt{\frac{2E_M}{SNR}}y_{1I}^{\prime\prime}-\textbf{R}(1,2)x_{2I}}{\textbf{R}(1,1)} 
\end{equation*}
and $\check{\textbf{x}}_I$ is given by that $(x_{1I}, x_{2I})$ pair that minimizes 
\begin{equation*}
f(\textbf{x}_I) = \left \Vert \textbf{y}_{I}^{\prime\prime} - \sqrt{\frac{SNR}{2E_M}}\textbf{Rx}_{I} \right \Vert^2. 
\end{equation*}
So, there are only $\sqrt{M}$ searches (for $\sqrt{M}$ possibilities for $x_{2I}$) involved in minimizing the ML-metric. The operation shown on the R.H.S of \eqref{quantize} quantizes $x_{1I}$ to its nearest possible value for a fixed $x_{2I}$. This is made possible due to the structure of $M$-QAM which is a Cartesian product of two $\sqrt{M}$-PAM constellations. The same method can be applied to obtain $\check{\textbf{x}}_Q$. So, the ML-decoding complexity is $\mathcal{O}(\sqrt{M})$.\\
{\it Case 2:} $0 < \gamma \leq \gamma_2^\prime$.

From Table \ref{tableqam} and also as was pointed out earlier, for $0 < \gamma \leq \gamma_2^\prime$, $\psi^* = 0$ and $\theta^* = \tan^{-1}\left(\frac{1}{\sqrt{M}} \right)$. This means that transmission is made only on the first virtual subchannel and the received signal of interest, with regard to \eqref{revised_channel_model}, can be expressed as 
\begin{equation*}
 y^\prime_1 = ax^\prime + n^\prime_1,
 \end{equation*}
where $n^\prime_1$ is the first element of $\textbf{n}^\prime$, $a = \sqrt{\sigma_1^2 SNR/((M+1)E_M)}$ and $x^\prime = \sqrt{M}x_1+x_2$, where $x_1$ and $x_2$ take values from $\mathcal{A}_{M-QAM}$. From Lemma \ref{lemma4}, $x^\prime$ takes values\footnote[5]{From a bit error rate point of view, it is advisable to transmit a symbol $x_1$ alone on the first virtual subchannel, with $x_1$ taking values from a Gray coded $M^2-$QAM. This is because the constellation given by \eqref{eq_constellation} will not be Gray coded. However, with a view of minimizing the word error rate, transmission of $x^\prime = \sqrt{M}x_1+x_2$, with $x_1$ and $x_2$ taking values from $M-$QAM, is as good a strategy as transmitting  $x_1$ alone, with $x_1$ taking values from a Gray coded $M^2-$QAM.} from $\mathcal{A}_{M^2-QAM}$. So, in the first step, $x^\prime$ is decoded to obtain $\check{x}^\prime = \check{x}_I^\prime + j\check{x}_Q^\prime$ by quantizing, where $\check{x}_I^\prime$ and $\check{x}_Q^\prime$ are given by

\begin{eqnarray*}
 \check{x}^\prime_{I} & = &\min\left(\max\left(2.\operatorname{rnd}\left[\frac{\frac{y_{1I}^\prime}{a}+1}{2} \right]-1,-M+1\right),M-1 \right),\\
\check{x}^\prime_{Q} & = &\min\left(\max\left(2.\operatorname{rnd}\left[\frac{\frac{y_{1Q}^\prime}{a}+1}{2} \right]-1,-M+1\right),M-1 \right).
\end{eqnarray*}

\noindent From $\check{x}^\prime$, $x_1$ is decoded to obtain $\check{x}_{1} = \check{x}_{1I} +j\check{x}_{1Q}$, with $\check{x}_{1I}$ and $\check{x}_{1Q}$ given by
\begin{equation}\label{op1}
 \check{x}_{1I}  = \operatorname{sgn}(\check{x}^\prime_{I})\left(2\left\lceil\frac{\vert \check{x}^\prime_{I}\vert}{2\sqrt{M}} \right\rceil -1 \right), ~~
\check{x}_{1Q} =  \operatorname{sgn}(\check{x}^\prime_{Q})\left(2\left\lceil\frac{\vert \check{x}^\prime_{Q}\vert}{2\sqrt{M}} \right\rceil -1 \right)
\end{equation}
and $x_2$ is decoded to obtain $\check{x}_{2} = \check{x}_{2I} +j\check{x}_{2Q}$, with $\check{x}_{2I}$ and $\check{x}_{2Q}$ given by 
\begin{equation}\label{op2}
 \check{x}_{2I} =  \check{x}_{I}^\prime - \sqrt{M}\check{x}_{1I}, ~~ \check{x}_{2Q} =  \check{x}_{Q}^\prime - \sqrt{M}\check{x}_{1Q}.
\end{equation}
 Note that the operations shown in \eqref{op1} and \eqref{op2} together perform the inverse of the function given by
\begin{equation*}
f(\check{x}_{1I},\check{x}_{1Q},\check{x}_{2I},\check{x}_{2Q}) = \sqrt{M}(\check{x}_{1I}+j\check{x}_{1Q}) + (\check{x}_{2I}+j\check{x}_{2Q})
\end{equation*}
\noindent for $\check{x}_{1I},\check{x}_{1Q},\check{x}_{2I},\check{x}_{2Q} \in A_{\sqrt{M}-PAM}$. Therefore, decoding $x_1$ and $x_2$ requires no exhaustive search over the $M$ signal points of the constellation. 
\end{proof}

It has to be pointed out that the advantage of not having to search over any of the signal points when $0 < \gamma \leq \gamma_2^\prime$ is unique to the proposed real-valued precoder and not obtainable for the case of the complex-valued optimal precoder \cite{optimal_precoder} for 4-QAM, for which the effective constellation when $0 < \gamma \leq 0.3016$ appears like a $\pi/12$ rotated QAM constellation (it is not exactly a rotated QAM constellation, however. Hence, when, $0 < \gamma \leq 0.3016$, even the sphere decoder cannot be used, since the effective constellation is not a lattice).

\section{Extension for $n_t > 2$} \label{sec_extension}
For the case of two transmit antennas, it is possible to obtain SVD-based, approximately optimal precoders (complex-valued precoder for $4$-QAM, real-valued precoder for any $M$-QAM). Such precoders are defined by two or three parameters, depending on whether the precoder is real-valued or complex-valued, respectively. However, such an approach cannot be taken for the case of $n_t > 2$, since, even for $n_t =3$, an optimal precoder would be defined by as many as 5 parameters, ruling out the possibility of a computer search even for $4$-QAM. So, a more practical way of obtaining a precoder with a reasonable error performance is to pair the $i^{th}$ and the $(n_{min}-i+1)^{th}$ subchannels along with the $i^{th}$ and the $(n_{min}-i+1)^{th}$ symbols, $i = 1,2\cdots,n_{min}/2$ and use the precoding scheme for 2 transmit antennas for this pair. This method of pairing has been shown to be the best in \cite{edmin_precoder} and has also been adopted in \cite{xy_codes}. The precoder would then have an '$X$' structure, as in \eqref{edmin_structure}. For the $i^{th}$ subchannel pairing, $\gamma_i \triangleq \tan^{-1}(\sigma_{n_{min}-i+1}/\sigma_i)$, $\rho_i = \sqrt{\sigma_i^2+\sigma_{n_{min}-i+1}^2}$ and  
\begin{equation}\label{delmin_extended}
 \delta(\gamma_i,\mathcal{A}_{M-QAM}) = \sin^2\gamma_i\left(\frac{c_i + A_{ki}d_{i}}{A_{ki} + \tan^2\gamma_i}\right), 	
\end{equation}
where $c_i,b_i$ and $A_{ki}$ are as defined in the previous section without the subscript $i$ (refer to \eqref{explicit_delmin}) and depend on $\gamma_i$ and $M$. Proceeding on the lines of the proof of Theorem \ref{lemma3}, the instantaneous WEP $P_{e,\textbf{D}}$ is upper bounded as 

\begin{equation}\label{min_wep}
 P_{e,\textbf{D}} \leq  \left(M ^{n_{min}}-1\right)Q\left(\sqrt{\frac{SNR}{2n_tE_M}}d_{min}\right),
\end{equation}
where $E_M = 2(M-1)/3$ and $d_{min}  =   \min_{\Delta \textbf{x}, \Delta\textbf{x} \neq \textbf{O}_{n_{min}\times1}} \Vert \textbf{DP}\Delta \textbf{x}\Vert$, with $\Delta \textbf{x} \in $ $\left\{ \textbf{x-x}^\prime  |  \textbf{x,x}^\prime \in \mathcal{A}_{M-QAM}^{n_{min}\times1}  \right\}$ and $\textbf{P}$ being the precoder with the $i^{th}$ and $n_{min}-i+1^{th}$ subchannels paired using the proposed precoding scheme described in Section \ref{sec_proposed}. So, 
\begin{eqnarray*}
d_{min} & = &\min_{i}\left\{\rho_i\sqrt{\frac{2n_t\delta(\gamma_i,\mathcal{A}_{M-QAM})}{n_{min}}}\right\},
\end{eqnarray*}
\noindent with $\delta(\gamma_i,\mathcal{A}_{M-QAM})$ given by \eqref{delmin_extended}. Observe that the scaling factor of $2n_t/n_{min}$ has been used to take into account the constraint that $\Vert \textbf{P} \Vert^2 = n_t$. Since the values of $\delta(\gamma_i,\mathcal{A}_{M-QAM})$ are known, we can enhance the error performance of the precoder by pre-multiplying the precoding matrix with a power control matrix ${\bf \Upsilon} = \operatorname{diag}(\tau_1,\tau_2,\cdots, \tau_{n_{min}/2},\tau_{n_{min}/2},\cdots,\tau_2,\tau_1)$ such that
\begin{equation}\label{equal_dmin}
\tau_i^2 \rho_i^2\delta(\gamma_i,\mathcal{A}_{M-QAM}) = \eta^2,  ~~~\forall i \in \left\{ 1,2,\cdots,\frac{n_{min}}{2}\right\},
\end{equation}

\noindent where $\eta$ is a constant and the power constraint on ${\bf \Upsilon}$ is such that $\Vert {\bf \Upsilon} \Vert^2 = 2\sum_{i=1}^{n_{min}/2} \tau_i^2 = n_{min}$. Due to this power constraint, from \eqref{equal_dmin}, we obtain

\begin{equation*}
 \tau_i = \sqrt{\frac{n_{min}} {2\rho_i^2\delta(\gamma_i,\mathcal{A}_{M-QAM})}\left(\sum_{j=1}^{n_{min}/2}\frac{1}{\rho_j^2\delta(\gamma_j,\mathcal{A}_{M-QAM})}\right)^{-1 }},
\end{equation*}
where $\delta(\gamma_i,\mathcal{A}_{M-QAM})$ is obtainable from \eqref{delmin_extended}. Hence, the proposed precoder has the structure given in \eqref{edmin_structure}, where
\begin{equation*}
 \textbf{M}_i = \sqrt{\frac{2n_t\tau_i^2}{n_{min}}}\left[\begin{array}{rr}
   \cos\psi_i\cos\theta_i     &    -\cos\psi_i\sin\theta_i \\
   \sin\psi_i\sin\theta_i     &    \sin\psi_i\cos\theta_i \\       
                      \end{array}\right],
\end{equation*}
with $\psi_i$ and $\theta_i$ being the approximately optimal values obtainable from \eqref{theta_eq} and \eqref{tan_gammatan_phi}, respectively, both depending on $\gamma_i$ and $M$. For example, for a $4\times4$ system using $4$-QAM signalling, if, for some channel realization, $\gamma_1 = \tan^{-1}\left(\frac{\sigma_4}{\sigma_1}\right) = \frac{1}{4}$ and $\gamma_2 = \tan^{-1}\left(\frac{\sigma_3}{\sigma_2}\right) = \tan^{-1}\left(\sqrt{\frac{2}{3}}\right)$, then, from Table \ref{tableqam}, $\theta_1 = \tan^{-1}(1/2)$, $\psi_1 = 0$, $\theta_2 = \pi/4$, $\psi_2 = \tan^{-1}\left(\frac{1}{\sqrt{3}\tan(\gamma_2)}\right)$
and $\tau_1 = \sqrt{\frac{5\eta^2}{\left(\sigma_1^2+\sigma_4^2\right)\cos^2\gamma_1}}$, $\tau_2 = \sqrt{\frac{\left(1+3\tan^2\gamma_2\right)\eta^2}{2\left(\sigma_2^2+\sigma_3^2\right)\sin^2\gamma_2}}$, $\eta^2 = 2\left( \frac{5}{\left(\sigma_1^2+\sigma_4^2\right)\cos^2\gamma_1} + \frac{1+3\tan^2\gamma_2}{2\left(\sigma_2^2+\sigma_3^2\right)\sin^2\gamma_2}    \right)^{-1}$.

The upper bound on the instantaneous WEP is now given as 
\begin{equation}\label{min_wep1}
 P_{e,\textbf{D}} \leq  \left(M ^{n_{min}}-1\right)Q\left(\sqrt{\frac{SNR }{n_{min}E_M}}\eta\right),
\end{equation}
where 
\begin{equation}\label{eta_rho}
 \eta = \sqrt{\frac{n_{min}} {2}\left(\sum_{j=1}^{n_{min}/2}\frac{1}{\rho_j^2\delta(\gamma_j,\mathcal{A}_{M-QAM})}\right)^{-1 }}.
\end{equation}
It can easily be checked that 
\begin{equation*}
Q\left(\sqrt{\frac{SNR}{2n_tE_M}}d_{min}\right) \geq Q\left(\sqrt{\frac{SNR }{n_{min}E_M}}\eta\right)
\end{equation*}
and hence, the upper bound in \eqref{min_wep1} is lower than that in \eqref{min_wep}. Therefore, the use of the power control matrix enhances error performance. Note that the symbols of each subsystem can be decoded independently from the symbols of the other subsystems. Hence, the ML-decoding complexity offered by our precoding scheme is $\mathcal{O}(\sqrt{M})$.

A similar approach of using a power control matrix has been taken in \cite{edmin_precoder} for $4$-QAM, but since we need to have explicit values of $\delta(\gamma_i,\mathcal{A}_{M-QAM})$, applying this scheme for the E-$d_{min}$ precoder with larger constellations is not feasible. Structurally, the E-$d_{min}$ precoder and the $X$-precoder differ from \eqref{edmin_structure} in that for the E-$d_{min}$ precoder, $\textbf{M}_i$ is optimized using an additional parameter $\phi_i$ (as shown in \eqref{opt_complex}), while for the $X$-precoder, $\textbf{M}_i$ is optimized with $\tau_i =1$ and $\psi_i = \pi/4$. Table \ref{comp_table} gives a comparison of the various low ML-decoding complexity precoding schemes. 

\section{Diversity gain}\label{sec_div_gain}
The E-$d_{min}$ precoder, the $X$-precoder and the $Y$-precoder have all been shown to guarantee a diversity gain equal to $\left(n_t-\frac{n_{min}}{2}+1\right)\left(n_r-\frac{n_{min}}{2}+1\right)$. Recall that the condition in Theorem \ref{lemma3} is only a sufficient condition for achieving full-diversity gain equal to $n_tn_r$. It is not necessary that $\textbf{P}$ be such that $[\textbf{P}\Delta \textbf{x} ]_1 \neq 0$, for $\Delta \textbf{x} \in \{\textbf{x-x}^\prime, ~\textbf{x,x}^\prime \in \mathcal{A}^{n_{min} \times 1} \}$. This can be seen by noting that for $n_t = 2$ and $4$-QAM, our precoder does not satisfy the condition when $\theta^* = \pi/4$, but still gives full-diversity. This is proved in the following lemma.
\begin{lemma}\label{lemma5}
 The proposed precoder offers full-diversity, i.e., a diversity gain equal to $2n_r$ for $n_t =2$.
\end{lemma}

\begin{proof}
 Consider the precoder given by
\begin{equation*}
 \textbf{P} = \left[\begin{array}{rr}
                     \cos\left(0.5\tan^{-1}2\right) & -\sin\left(0.5\tan^{-1}2\right) \\
\sin\left(0.5\tan^{-1}2\right) & \cos\left(0.5\tan^{-1}2\right)\\
                    \end{array}\right],
\end{equation*}
which is the full-diversity rotation matrix \cite{full_div} in 2 dimensions and has the highest non-zero product distance among all $2\times2$ sized orthogonal matrices. This precoder, which we call the lattice precoder for $n_t=2$, has full-diversity from Theorem \ref{lemma3}. Clearly, $\delta(\gamma,\mathcal{A}_{M-QAM})$ for any value of $\gamma$ is greater for our precoder than that for the lattice precoder, since our precoder is approximately optimal among real-valued precoders. So, our precoder has better error performance than the lattice precoder. Hence, our precoder too offers full-diversity, like the lattice precoder for $n_t=2$. 
\end{proof}

From Lemma \ref{lemma5}, Theorem \ref{lemma2} and Theorem \ref{lemma3}, one would be inclined to believe that for $n_t >2$, a subsystem with index
$i$, for which the $i^{th}$ and the $n_{min}-i+1^{th}$ virtual subchannels are paired and the $i^{th}$ and the $n_{min}-i+1^{th}$ symbols precoded by the scheme  proposed in Section \ref{sec_proposed}, has a diversity gain of $(n_t-i+1)(n_r-i+1)$, with $i = 1,2,\cdots,n_{min}/2$, in which case the diversity gain of the whole system would be the minimum of the diversity gains of all the subsystems, i.e., $(n_t-n_{min}/2+1)(n_r-n_{min}/2+1)$. In fact, the diversity gains of systems using the E-$d_{min}$ precoder and the $X$-precoder have been claimed to be $(n_t-n_{min}/2+1)(n_r-n_{min}/2+1)$ due to this reason. It must be noted that the power control matrix ${\bf \Upsilon}$ plays an important role in the error performance of our precoder (also the E-$d_{min}$ precoder for $4$-QAM), as explained in Section \ref{sec_extension}. Before we analyze the achievable diversity gain of the system with the proposed precoding scheme, the following important observation needs to be made about $\delta(\gamma_i, \mathcal{A}_{M-QAM})$. Since $\sigma_1 \geq \sigma_2 \geq \cdots \sigma_{n_{min}}$, we have  $ \frac{\sigma_{n_{min}}}{\sigma_1} \leq \frac{\sigma_{n_{min}-1}}{\sigma_2} \leq \cdots \frac{\sigma_{\frac{n_{min}}{2}+1}}{\sigma_{\frac{n_{min}}{2}}}$. Consequently, 
\begin{equation*}
 \tan^{-1}\left(\frac{\sigma_{n_{min}}}{\sigma_1}\right) \leq \tan^{-1}\left(\frac{\sigma_{n_{min}-1}}{\sigma_2}\right) \leq \cdots \tan^{-1}\left(\frac{\sigma_{\frac{n_{min}}{2}+1}}{\sigma_{\frac{n_{min}}{2}}}\right)
\end{equation*}
and therefore, $\gamma_1 \leq \gamma_2 \cdots \leq \gamma_{n_{min}/2}$. From Fig. \ref{min_dist_plots}, except for the case of $4$-QAM, we can conclude that $
\delta(\gamma_i,\mathcal{A}_{M-QAM}) \leq \delta(\gamma_j,\mathcal{A}_{M-QAM})$, for $i<j$. Due to this fact, although it is expected that for $1 \leq i < j \leq n_{min}/2$, $\rho_i^2 \geq \rho_j^2$, it is not guaranteed that $\rho_i^2\delta(\gamma_i,\mathcal{A}_{M-QAM}) > \rho_j^2\delta(\gamma_j,\mathcal{A}_{M-QAM})$, due to which even without the use of ${\bf \Upsilon}$, the overall diversity gain of the system might be higher than $(n_t-n_{min}/2+1)(n_r-n_{min}/2+1)$ (this holds true even for the $X$-precoder). With the use of ${\bf \Upsilon}$ for our proposed precoder, the channel dependent instantaneous WEP is dependent on $\eta$, as seen in \eqref{min_wep1}. Let
\begin{equation*}
 \zeta \triangleq \frac{\eta}{\rho_1\sqrt{\delta(\gamma_1,\mathcal{A}_{M-QAM})}}
\end{equation*}
and $P_{\zeta}$ be the probability that $\zeta < 1$.

In Table \ref{table_div}, we tabulate the values of $\zeta_{min}$, which is the minimum value of $\zeta$ obtained on simulations for $10^7$ channel realizations, and $P_{\zeta}$, which is again calculated by simulating $10^7$ channel realizations, for different MIMO systems. In the table, we observe that for $n_t =16,32$ and for $M\geq 64$, $\zeta$ is always greater that 1. This can be attributed to the fact that for higher values of $n_{min}$, the ratio of $\sigma_{n_{min}}$ to $\sigma_1$ is very low and the corresponding value of $\delta(\gamma_1,\mathcal{A}_{M-QAM})$ is also very low. For such systems, we can safely say that the full-diversity gain equal to $n_tn_r$ is achieved (since $\sigma_1^2$ is associated with a diversity gain of $n_tn_r$). For other systems, the simulations results in Table \ref{table_div} seem to indicate that there exists a $\zeta_{min} > 0$ such that $\zeta_{min}\leq \zeta$, i.e., $\zeta$ is lower bounded by $\zeta_{min}$. So, from \eqref{min_wep1}, 
\begin{equation*}
  P_{e,\textbf{D}} \leq  \left(M ^{n_{min}}-1\right)Q\left(\sqrt{\frac{SNR.\zeta_{min}^2\rho_1^2\delta(\gamma_1,\mathcal{A}_{M-QAM})}{n_{min}E_M}}\right).
\end{equation*}
Let $\delta_M = \min\{\delta(\gamma_1,\mathcal{A}_{M-QAM})\}$, which is a constant depending on $M$. Then, 
\begin{eqnarray*}
 P_{e,\textbf{D}} & \leq & \left(M ^{n_{min}}-1\right)Q\left(\sqrt{\frac{SNR.\zeta_{min}^2\rho_1^2\delta_M}{n_{min}E_M}}\right)\leq \left(M ^{n_{min}}-1\right)Q\left(\sqrt{\frac{SNR.\zeta_{min}^2\sigma_1^2\delta_M}{n_{min}E_M}}\right)\\
& = & \left(M ^{n_{min}}-1\right)Q\left(\sqrt{\left(\frac{\zeta_{min}^2\delta_M}{n_{min}E_M}\right)\lambda_1SNR}\right),
\end{eqnarray*} 
where, as used throughout the paper, $\lambda_1 = \sigma_1^2$. From Theorem \ref{lemma1}, we obtain, as $SNR \to \infty$,
\begin{equation*}
 P_e \leq C.SNR^{-n_tn_r} + o\left(SNR^{-n_tn_r}\right).
\end{equation*}
where 
\begin{equation}\label{coding_gain}
 C =  \left(M ^{n_{min}}-1\right)\frac{a_1(2n_tn_r-1)(2n_tn_r-3)\cdots1}{2n_tn_r}\left( \frac{\delta_M \zeta_{min}^2}{n_{min}E_M}\right)^{-n_tn_r},
\end{equation}
 with $a_1$ being a constant in the expression for the marginal PDF of $\lambda_1$, as defined in Theorem \ref{lemma2}. Therefore, the overall diversity gain of the system is $n_tn_r$. Note that in \eqref{coding_gain}, $\delta_M$ and $\zeta_{min}$ define the coding gain - the higher the value of $\zeta_{min}$ and $\delta_M$, the better the error performance. It is not known if $\zeta_{min}$ can be obtained analytically. The values in Table \ref{table_div} are only indicative of what the actual $\zeta_{min}$ is likely to be. For example, for the $16 \times16$ system with $64$-QAM, $\zeta_{min}$ is likely to be greater than 1. Thus, we have shown that our precoding scheme provides full-diversity. This claim is supported by the WEP plots for different MIMO systems, shown in the following section.

\section{Simulation results}\label{sec_simulations}
For all simulations, we consider the Rayleigh fading channel with prefect CSIT and CSIR. We consider three MIMO systems - $2\times2$, $4\times4$ and $8\times8$ MIMO systems. For the $2\times2$ MIMO system, the rival precoders for our precoder are the E-$d_{min}$ precoder and the $X$-precoder. We have left out the $Y$-precoder since it has been shown in \cite{xy_codes} to have an error performance comparable with that of the $X$-precoder for $4$-QAM, while for $16$-QAM, it is not expected to beat the $X$-precoder, as can be inferred from Fig. \ref{min_dist16}. The constellations employed are $4$-QAM and $16$-QAM. For $16$-QAM, the E-$d_{min}$ precoder is not considered since it is very hard to obtain and not explicitly stated in literature. For the $X$-precoder, we have obtained the approximately optimal angles for $16$-QAM using a numerical search for $\gamma = k.\varDelta$, $k = 1,2,\cdots, \lfloor \frac{\pi}{4\varDelta} \rfloor$, $\varDelta = 0.001$, and have used a look-up table to obtain the appropriate angle for the corresponding value of $\gamma$ during simulations. A look-up table is necessary since the approximately optimal angle for the $X$-precoder is not a weighted sum of shifted step functions like that for our precoder. Fig. \ref{fig_cer2x2} shows the plots of the word error probability (WEP) as a function of the average SNR at each receive antenna for the $2\times2$ system. As expected, the E-$d_{min}$ precoder has the best error performance for $4$-QAM, marginally beating our precoder, which in turn significantly beats the $X$-precoder. For $16$-QAM, our precoder beats the $X$-precoder by about $1.5$dB at an SNR of $30$dB. 

For $4\times4$ and $8\times8$ systems, we also consider the Lattice precoder, which is the orthogonal matrix with the largest known non-zero product distance for $n_{min} = n_t$ real dimensions, and given explicitly in \cite{full_div}. This precoder has been shown in Theorem \ref{lemma3} to offer full-diversity. The plots of the WEP for the $4\times4$ system and the $8\times8$ system are given in Fig. \ref{fig_cer4x4} and Fig. \ref{fig_cer8x8}, respectively. The plots indicate that the E-$d_{min}$ precoder and our proposed precoder offer full-diversity, since they beat the full-diversity achieving Lattice precoder (even the $X$-precoder appears to offer full-diversity, losing out in coding gain only. The explanation for this has already been given in Section \ref{sec_div_gain}). Our precoder significantly outperforms the $X$-precoder while having lower expected ML-decoding complexity (as shown in Theorem \ref{thm2}), while the E-$d_{min}$ precoder has the best error performance for $4$-QAM, marginally beating our precoder, but this is at the expense of ML-decoding complexity. In Table \ref{probabilities}, by simulating $10^6$ channel realizations, we have tabulated the probability that ML-decoding can be done without searching over any of the signal points for $4$- and $16$-QAM. It can be noted that for the $2\times2$ MIMO system with $4$-QAM, for more than $50\%$ of the channel realizations, no search over any of the signal points is required, while for the $4\times4$ and the $8\times8$ MIMO systems, half the number of subsystems do not require any search over the constellation points for more than $99\%$ of the channel realizations. This advantage, however, diminishes with the increase in constellation size. 

\section{Discussion}\label{sec_discussion}
For systems with full CSIT, we have proposed a real-valued precoder for $n_t=2$, which, for QAM constellations, is approximately optimal among all real-valued precoders based on the SVD of the channel matrix and has an expected ML-decoding complexity lower than $\mathcal{O}(\sqrt{M})$. The advantage of the proposed precoder over the E-$d_{min}$ precoder is that it is much easier to obtain for larger QAM constellations and it also has lower ML-decoding complexity, while the loss in error performance for $4$-QAM is only marginal. The proposed precoder handsomely beats the $X$-precoder in error performance while having lower expected ML-decoding complexity. A precoding scheme for $n_t >2$ is also given and this scheme is shown to offer full-diversity with QAM constellations. It would be interesting to design low ML-decoding complexity, full-rate, full-diversity precoders for more realistic scenarios, like for systems with imperfect CSIT or partial CSIT. 

\section*{Acknowledgements}
This work was partly supported by the DRDO-IISc program on Advanced Research in Mathematical Engineering, through research grants, and the INAE Chair Professorship to B. Sundar Rajan.

\newpage

\begin{table}
\begin{center}
\begin{tabular}{|c|c|c|c|c|} \hline
 $M$ & $\gamma$  &  $\theta^*$ & $\tan\gamma\tan\psi^*$ &  $(p,q)$  \\ \hline \hline
 \multirow{2}{*}{$4$} & $0 - \tan^{-1}\left( \frac{1}{\sqrt{7}}\right)$ & $\tan^{-1}\left( \frac{1}{2}\right)$ &  $0$ & $(0,1)$, $(1,1)$ \\ \cline{2-5}
& $\tan^{-1}\left( \frac{1}{\sqrt{7}}\right) - \frac{\pi}{4}$ & $\frac{\pi}{4}$ &  $\frac{1}{\sqrt{3}}$ & $(0,1)$, $(1,1)$, $(1,0)$\\ \hline

\multirow{4}{*}{$16$}  & $0 - 0.1018 $ & $\tan^{-1}\left( \frac{1}{4}\right)$ &  $0$ & $(0,1)$, $(1,3)$\\ \cline{2-5}
&$0.1018 - 0.1567$ & $0.3474$ & $0.1096$ & $(0,1)$, $(1,3)$, $(1,2)$ \\ \cline{2-5}
&$0.1567 - 0.3479$ & $0.4914$ &  $0.2277$ & $(0,1)$, $(1,1)$, $(1,2)$ \\ \cline{2-5}
& $0.3479 - \frac{\pi}{4}$ & $\frac{\pi}{4}$ &  $\frac{1}{\sqrt{3}}$ & $(0,1)$, $(1,1)$, $(1,0)$ \\ \hline

\multirow{7}{*}{64} & $0 - 0.0273 $ & $\tan^{-1}\left( \frac{1}{8}\right)$ &  $0$ & $(0,1)$, $(1,7)$ \\ \cline{2-5}
& $0.0273 - 0.0354$ & $0.5450$ &  $0.0335$ & $(1,2)$, $(2,3)$, $(3,5)$  \\ \cline{2-5}
& $0.0354 - 0.0415$ & $0.3766$ &  $0.0393$ & $(1,2)$, $(1,3)$, $(2,5)$ \\ \cline{2-5}
& $0.0415 - 0.0519$ & $0.6325$ &  $0.0433$ & $(1,1)$, $(2,3)$, $(3,4)$ \\ \cline{2-5}
& $0.0519 - 0.0735$ & $0.2640$ &  $0.0620$ & $(0,1)$, $(1,3)$, $(1,4)$  \\ \cline{2-5}
& $0.0735 - 0.0975$ & $0.5763$ &  $0.0872$ & $(1,1)$, $(1,2)$, $(2,3)$\\ \cline{2-5}
& $0.0975 - 0.1567$ & $0.3474$ &  $0.1096$ & $(0,1)$, $(1,3)$, $(1,2)$\\ \cline{2-5}
& $0.1567 - \frac{\pi}{4}$ & same as $16$-QAM &  same as $16$-QAM &  same as $16$-QAM\\ \hline

\multirow{5}{*}{256} & $0 - 0.0071$ & $\tan^{-1}\left( \frac{1}{16}\right)$ &  $0$ & $(0,1)$, $(1,15)$ \\ \cline{2-5}
& $0.0071 - 0.0139$ & $0.5103$ &  $0.0098$ & $(1,2)$, $(4,7)$, $(5,9)$  \\ \cline{2-5}
& $0.0139 - 0.0278$ & $0.1501$ &  $0.0197$ & $(0,1)$, $(1,6)$, $(1,7)$\\ \cline{2-5}
& $0.0278- 0.0494$ & $0.2114$ &  $0.0394$ & $(0,1)$, $(1,4)$, $(1,5)$ \\ \cline{2-5}
& $0.0494 - 0.0735$ & $0.2640$ &  $0.0620$ & $(0,1)$, $(1,3)$, $(1,4)$ \\ \cline{2-5}
&$0.0735-\frac{\pi}{4}$ & same as $64$-QAM &  same as $64$-QAM &  same as $64$-QAM \\ \hline
\multirow{6}{*}{1024} & $0 - 0.0018 $ & $\tan^{-1}\left( \frac{1}{32}\right)$  & $0$ & $(0,1)$, $(1,31)$\\ \cline{2-5}
& $0.0018 - 0.0027 $ & $0.1301$ &  $0.0022$ & $(1,8)$, $(2,15)$, $(3,23)$ \\ \cline{2-5}
& $0.0027 - 0.0042$ & $0.2300$ &  $0.0035$ & $(1,4)$, $(3,13)$, $(4,17)$ \\ \cline{2-5}
& $0.0042 - 0.0065$ & $0.7304$ &  $0.0053$  & $(1,1)$, $(8,9)$, $(9,10)$ \\ \cline{2-5}
& $0.0065 - 0.0086$ & $0.3509$ &  $0.0079$ & $(1,3)$, $(3,8)$, $(4,11)$ \\ \cline{2-5}
& $0.0086 - \frac{\pi}{4}$ & same as $256$-QAM &  same as $256$-QAM  &  same as $256$-QAM \\ \hline
\end{tabular}
\end{center}
\caption{Approximately optimal values of $\theta$ and $\psi$ for various QAM constellations}
\label{tableqam}
\end{table}

\newpage

\begin{table}
\centering
\begin{threeparttable}
\begin{tabular}{|c|c|c|c|} \hline
 \multirow{3}{*}{Precoder} &  ML-decoding  & \multirow{3}{*}{ Existence for $\vert \mathcal{A} \vert = M$} & \multirow{3}{*}{{ Error performance}} \\ 
&  complexity\tnote{{\bf \ddag}} &   & \\
&   for $\vert \mathcal{A}\vert=M$ &   &  \\ \hline \hline
\multirow{3}{*}{E-$d_{min}$ precoder} & \multirow{3}{*}{$\mathcal{O}\left(M\sqrt{M}\right)$} & \multirow{3}{*}{exists only for $4$-QAM}  & the best   \\ 
& &  & for $M=4$ among\\ 
& & & known precoders\\ \hline
\multirow{3}{*}{$X$-precoder} & \multirow{3}{*}{$\mathcal{O}\left(\sqrt{M}\right)$} &  not possible without &  worse than the E-$d_{min}$  \\ 
& &the use of a look-up\tnote{{ \bf \textyen}} & precoder and the \\
& & table for $M >4$& proposed precoder\\ \hline
\multirow{3}{*}{$Y$-precoder} & \multirow{3}{*}{$\mathcal{O}\left(1\right)$} & closed form expression  & better than $X$-precoders\\ 
& & exists for any $M$, with $\mathcal{A}$  & for ill-conditioned channels  \\
& & a 2-dimensional constellation \cite{xy_codes} & only \\ \hline
The proposed  & \multirow{2}{*}{$\mathcal{O}\left(\sqrt{M}\right)$} & easy to obtain  &  much better than \\
precoder & &for any $M$-QAM & $X$-,$Y$-precoders\\ \hline
\end{tabular}
\begin{tablenotes}
\item[{ \bf \ddag}] when $\gamma_1^\prime \leq \gamma \leq \gamma_2^\prime$, for the E-$d_{min}$ precoder, a full search over all the signal points is needed, while for
 the proposed precoder, no search over signal points is needed.
\item[{ \bf \textyen}] amounts to storing the near-optimal angle values for $\gamma = k\varDelta$, $k = 1,2,\cdots, \lfloor \frac{\pi}{4\varDelta} \rfloor$, where $\varDelta$ is a suitable step size. 
\end{tablenotes}
\end{threeparttable}
\caption{Comparison of low ML-decoding complexity Precoding schemes}
\label{comp_table}
\end{table}

{\footnotesize
\begin{table}
\centering
\begin{threeparttable}
\begin{tabular}{|c|c|c|c|c|c|c|c|c|c|c|} \hline
 MIMO  & \multicolumn{2}{c|}{$M=4$} & \multicolumn{2}{c|}{$M=16$} & \multicolumn{2}{c|}{$M=64$} & \multicolumn{2}{c|}{$M=256$} & \multicolumn{2}{c|}{$M=1024$}\\ \cline{2-11}
system & $\zeta_{min}$ & $P_\zeta$ & $\zeta_{min}$ & $P_\zeta$ & $\zeta_{min}$ & $P_\zeta$ & $\zeta_{min}$ & $P_\zeta$  & $\zeta_{min}$ & $P_\zeta$ \\ \hline \hline
$4\times4$ & $0.11$ & $0.79$ &$0.33$ &$0.10$ &$0.45$ & $0.09$& $0.40$ & $0.09$ & $0.45$ & $0.09$\\ \hline
$8\times8$ & $0.28$ & $0.99$ & $0.59$ &$ 0.01$ & $0.76$ & $10^{-3}$ & $0.77$ & $1.1\times10^{-3}$ & $0.75$ & $10^{-3}$ \\ \hline
$16\times16$ &$0.35$ & $1$& $0.72$& $9.1\times10^{-4}$ & $1.16$& $0$& $1.17$ &$0$ & $1.12$ & $0$\\ \hline
$32\times32$ &$0.41$ &$1$ &$0.84$ & $4\times10^{-3}$ &$1.87$ &$0$ &$1.75$ & $0$ & $1.9$ &  $0$ \\ \hline
\end{tabular}
\end{threeparttable}
\caption{Characteristics of $\zeta$ for different MIMO systems}
\label{table_div}
\end{table}
}

{\footnotesize
\begin{table}
\centering
\begin{threeparttable}
\begin{tabular}{|c|c|c|c|} \cline{1-1} \cline{3-4}
MIMO system & & $M=4$ & $M=16$ \\ 
  \hline \hline
$2\times2$ & $Pr\{\frac{\sigma_2}{\sigma_1} \leq \tan\gamma_2^\prime \}$\tnote{{ \bf \dag}} & $0.5780$ & $0.0612$ \\ \hline
\multirow{2}{*}{$4\times4$} & $Pr\{\frac{\sigma_4}{\sigma_1} \leq \tan\gamma_2^\prime \}$ & $0.9942$ & $0.3286$ \\ \cline{2-4}
& $Pr\{\frac{\sigma_3}{\sigma_2} \leq \tan\gamma_2^\prime \}$ & $0.0620$ & $8\times10^{-6}$ \\ \hline
\multirow{4}{*}{$8\times8$} & $Pr\{\frac{\sigma_8}{\sigma_1} \leq \tan\gamma_2^\prime \}$ & $1$ &$0.8582$ \\ \cline{2-4}
& $Pr\{\frac{\sigma_7}{\sigma_2} \leq \tan\gamma_2^\prime \}$ & $0.9969$ & $0.0112$ \\ \cline{2-4}
& $Pr\{\frac{\sigma_6}{\sigma_3} \leq \tan\gamma_2^\prime \}$ & $0.2179$ & $0$ \\ \cline{2-4}
& $Pr\{\frac{\sigma_5}{\sigma_4} \leq \tan\gamma_2^\prime \}$ & $4\times10^{-6}$ & $0$ \\ \hline
\end{tabular}
\begin{tablenotes}
\item[{ \bf \dag}] $\gamma_2^\prime = \tan^{-1}\left(\frac{1}{\sqrt{7}}\right)$ for $M=4$ and $\gamma_2^\prime \approx 0.1018$ for $M =16$.
\end{tablenotes}
\end{threeparttable}
\caption{Probability that no search is required for each subsystem of different MIMO systems}
\label{probabilities}
\end{table}
}

\begin{figure}
\centering
\includegraphics[width=5.5in,height=3.5in]{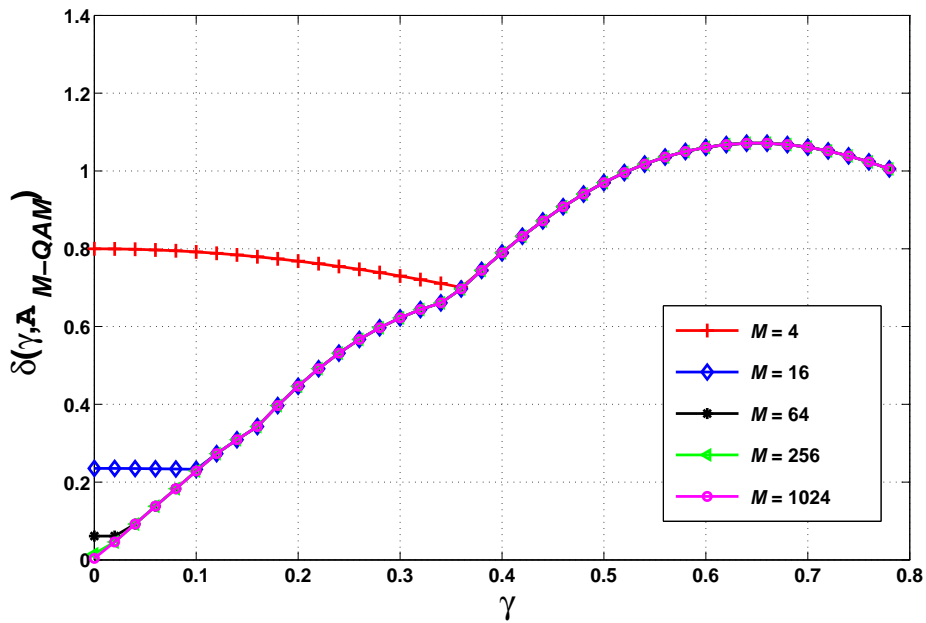}
\caption{$\delta(\gamma,\mathcal{A}_{M-QAM})$ as a function of $\gamma$ for the proposed precoder for various QAM constellations}
\label{min_dist_plots}
\end{figure}

\begin{figure}
\centering
\includegraphics[width=5.5in,height=3.5in]{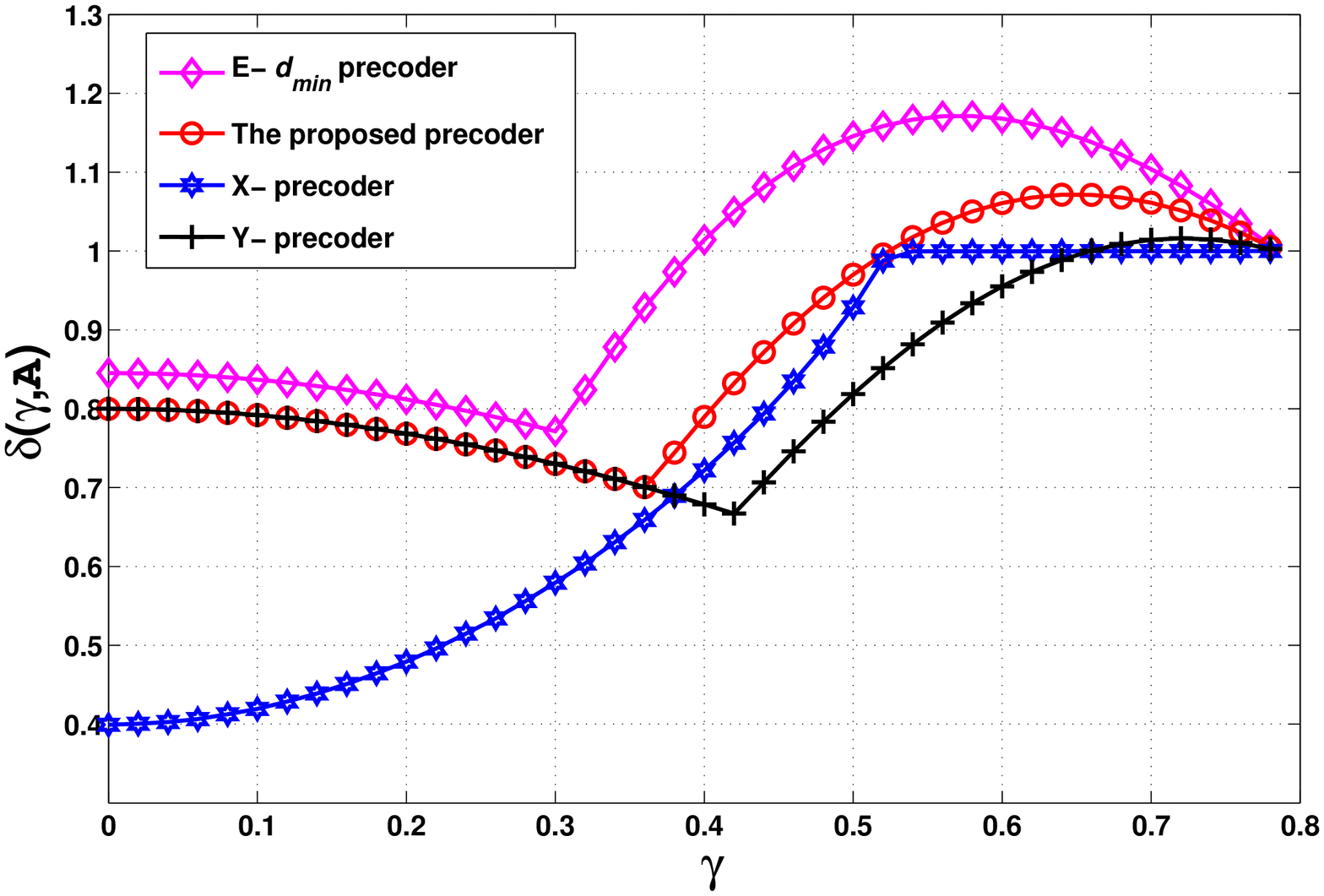}
\caption{$\delta(\gamma,\mathcal{A})$ comparison for $\vert \mathcal{A}\vert = 4$}
\label{min_dist4}
\end{figure}

\begin{figure}
\centering
\includegraphics[width=5.5in,height=3.5in]{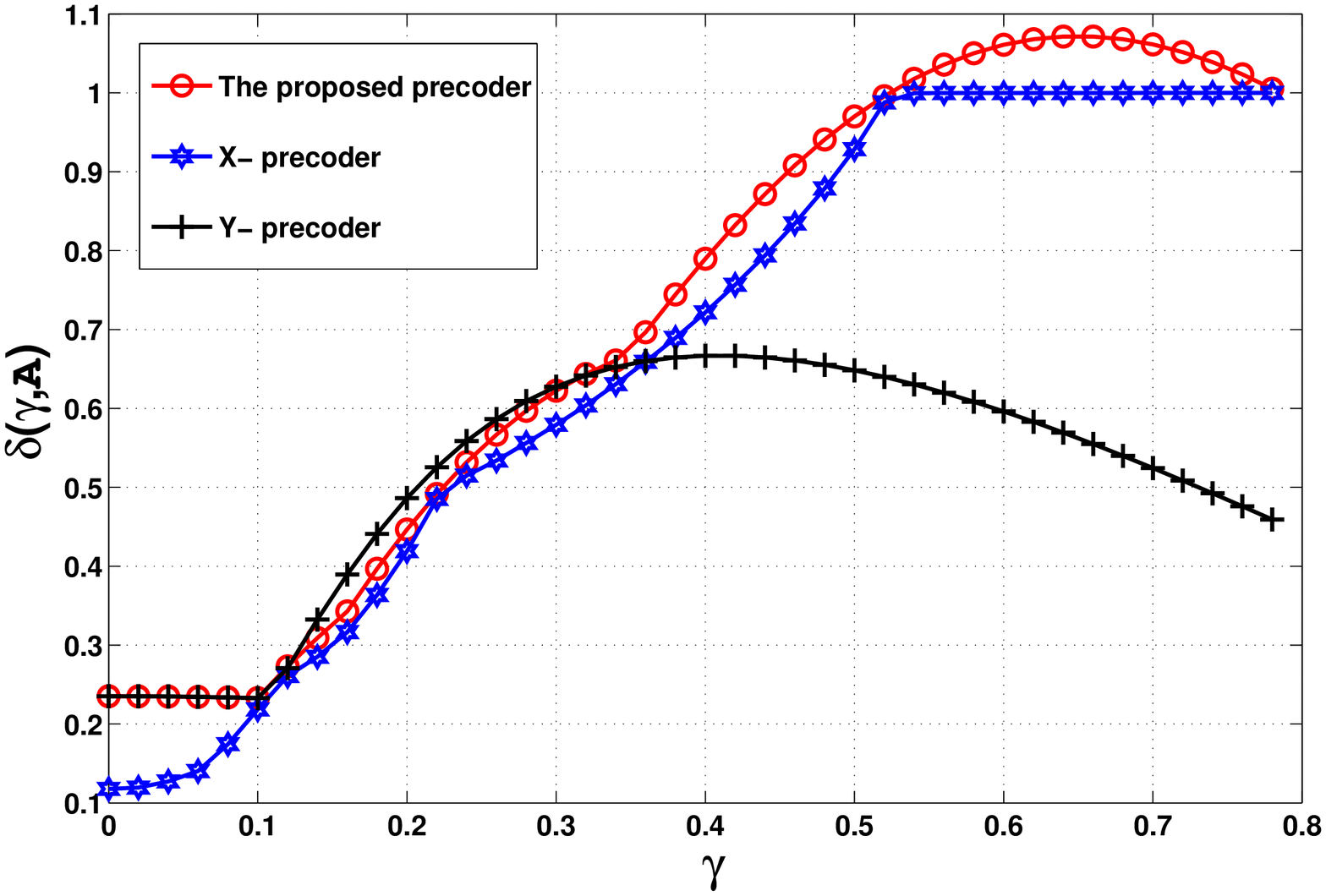}
\caption{$\delta(\gamma,\mathcal{A})$ comparison for $\vert \mathcal{A}\vert = 16$}
\label{min_dist16}
\end{figure}

\begin{figure}
\centering
\includegraphics[width=5.5in,height=3.5in]{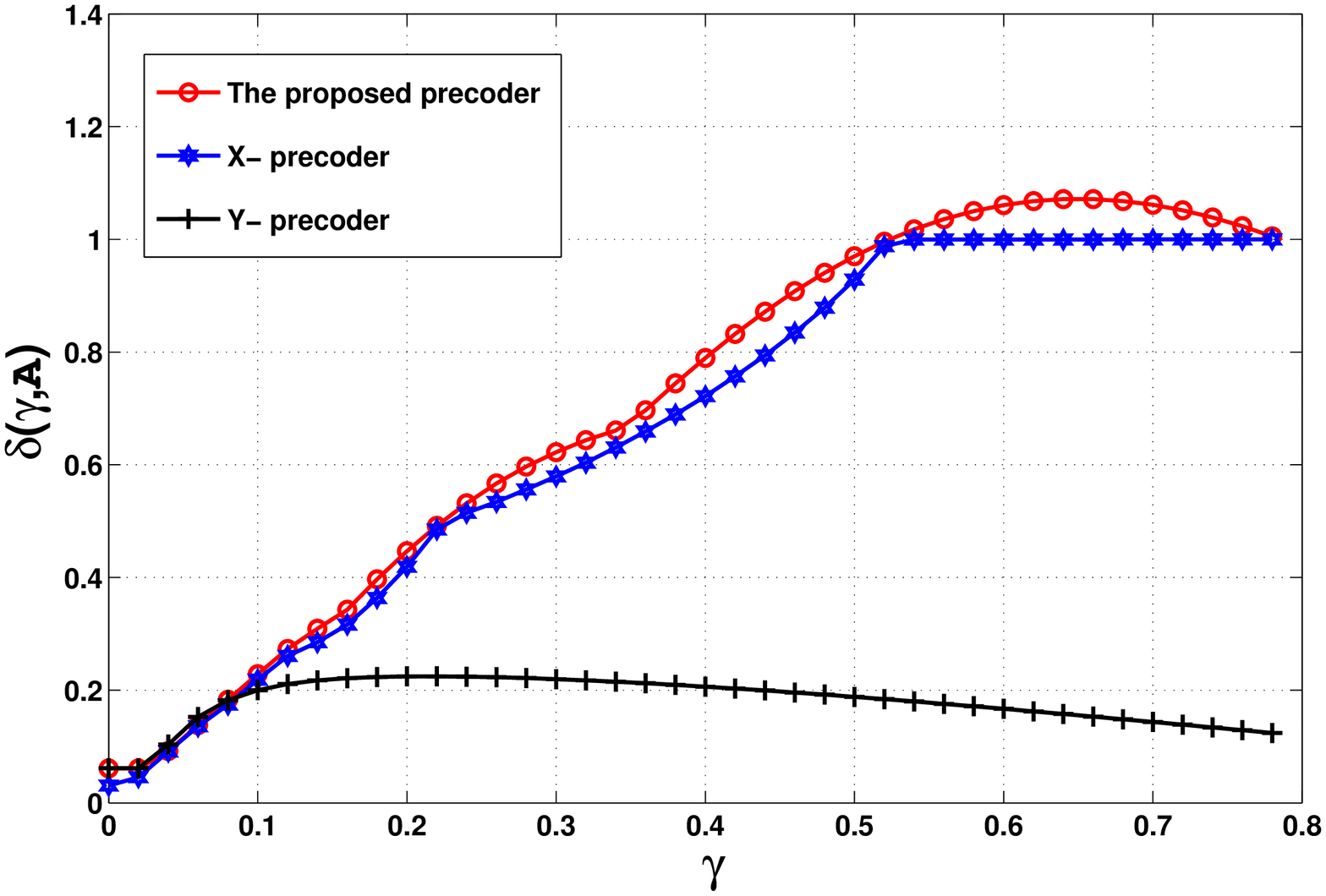}
\caption{$\delta(\gamma,\mathcal{A})$ comparison for $\vert \mathcal{A}\vert = 64$}
\label{min_dist64}
\end{figure}

\begin{figure}
\centering
\includegraphics[width=5.5in,height=3.5in]{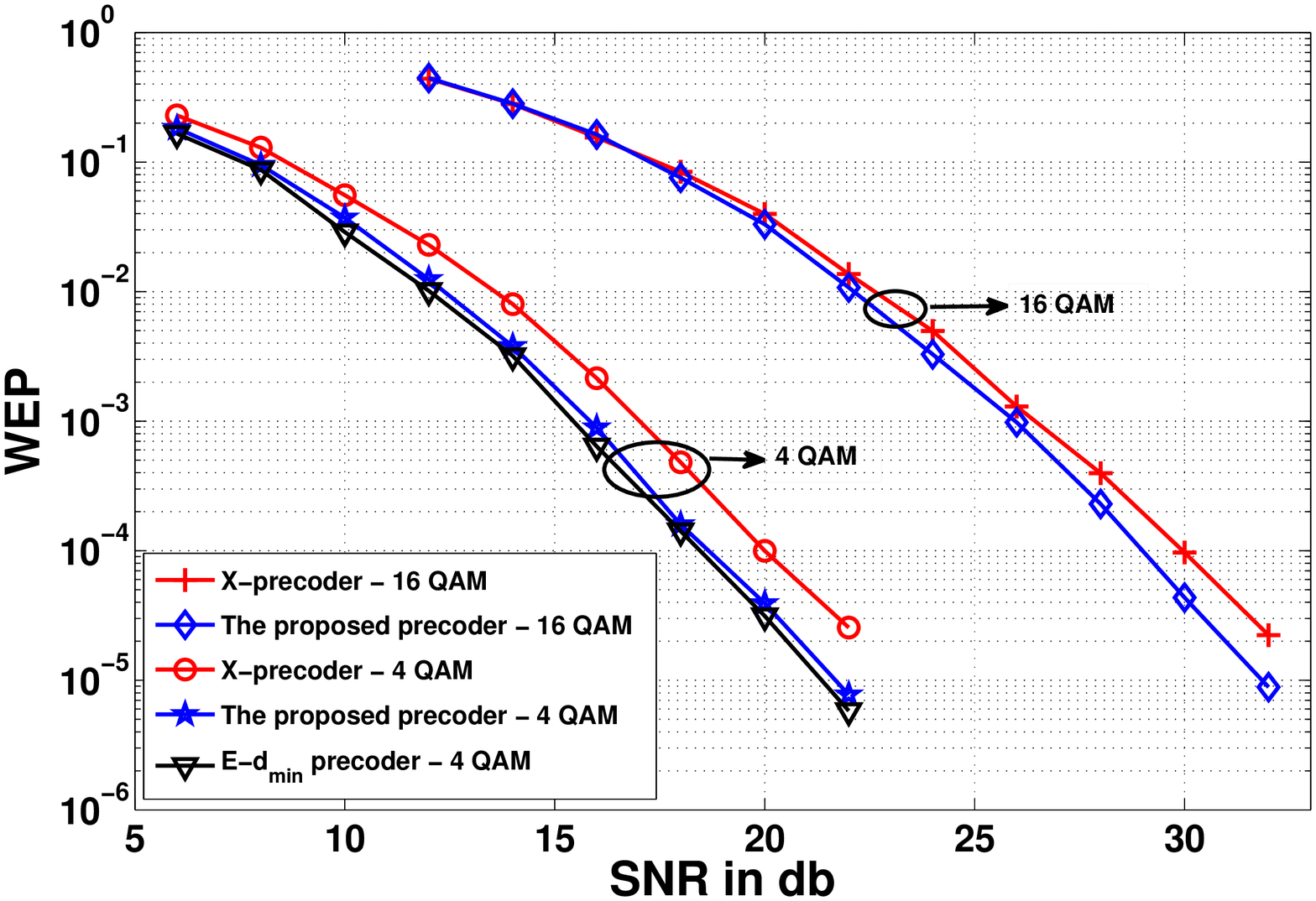}
\caption{WEP comparison for $2\times 2$ MIMO systems for $4$-QAM and $16$-QAM}
\label{fig_cer2x2}
\end{figure}

\begin{figure}
\centering
\includegraphics[width=5.5in,height=3.5in]{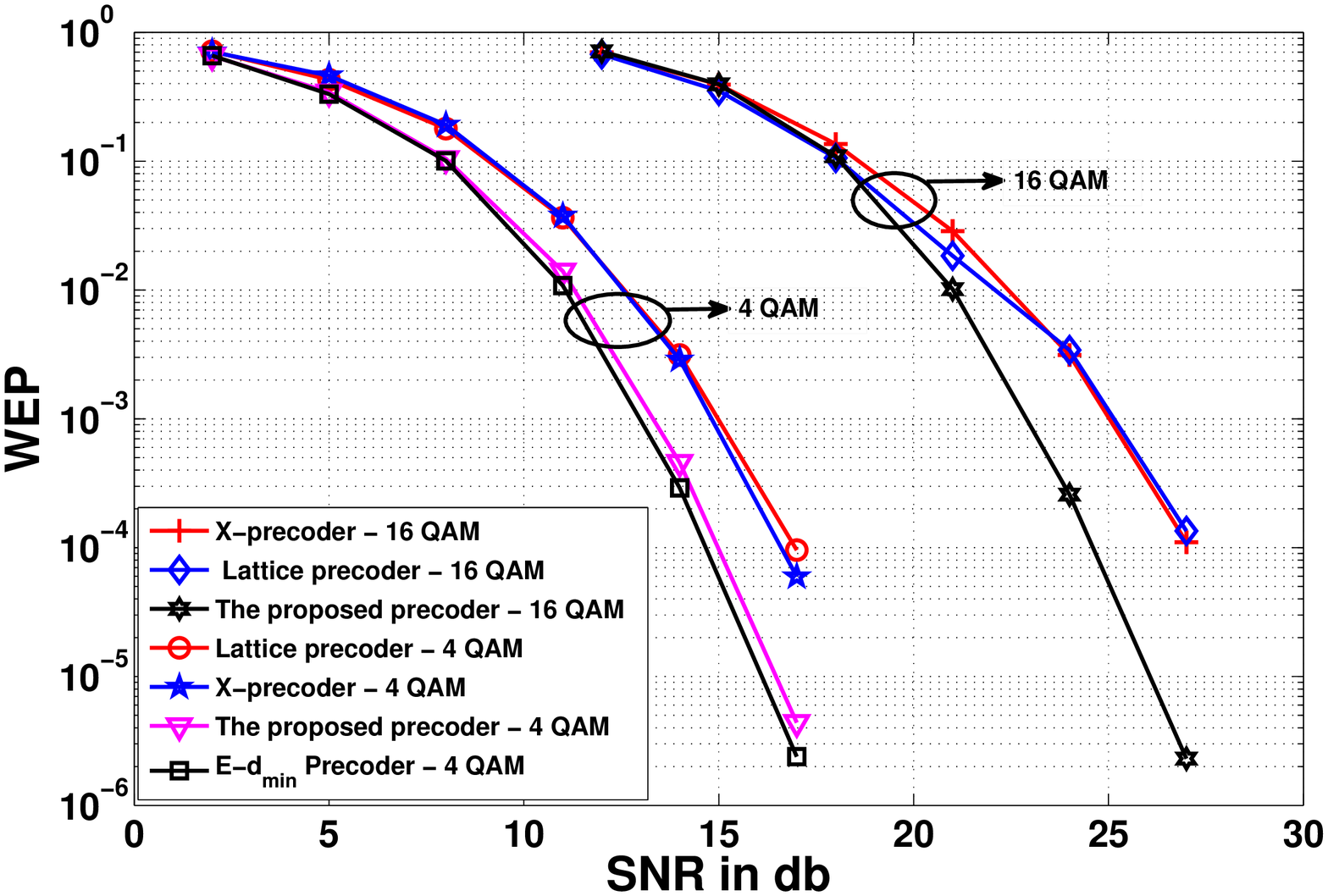}
\caption{WEP comparison for $4\times 4$ MIMO systems for $4$-QAM and $16$-QAM}
\label{fig_cer4x4}
\end{figure}

\begin{figure}
\centering
\includegraphics[width=5.5in,height=3.5in]{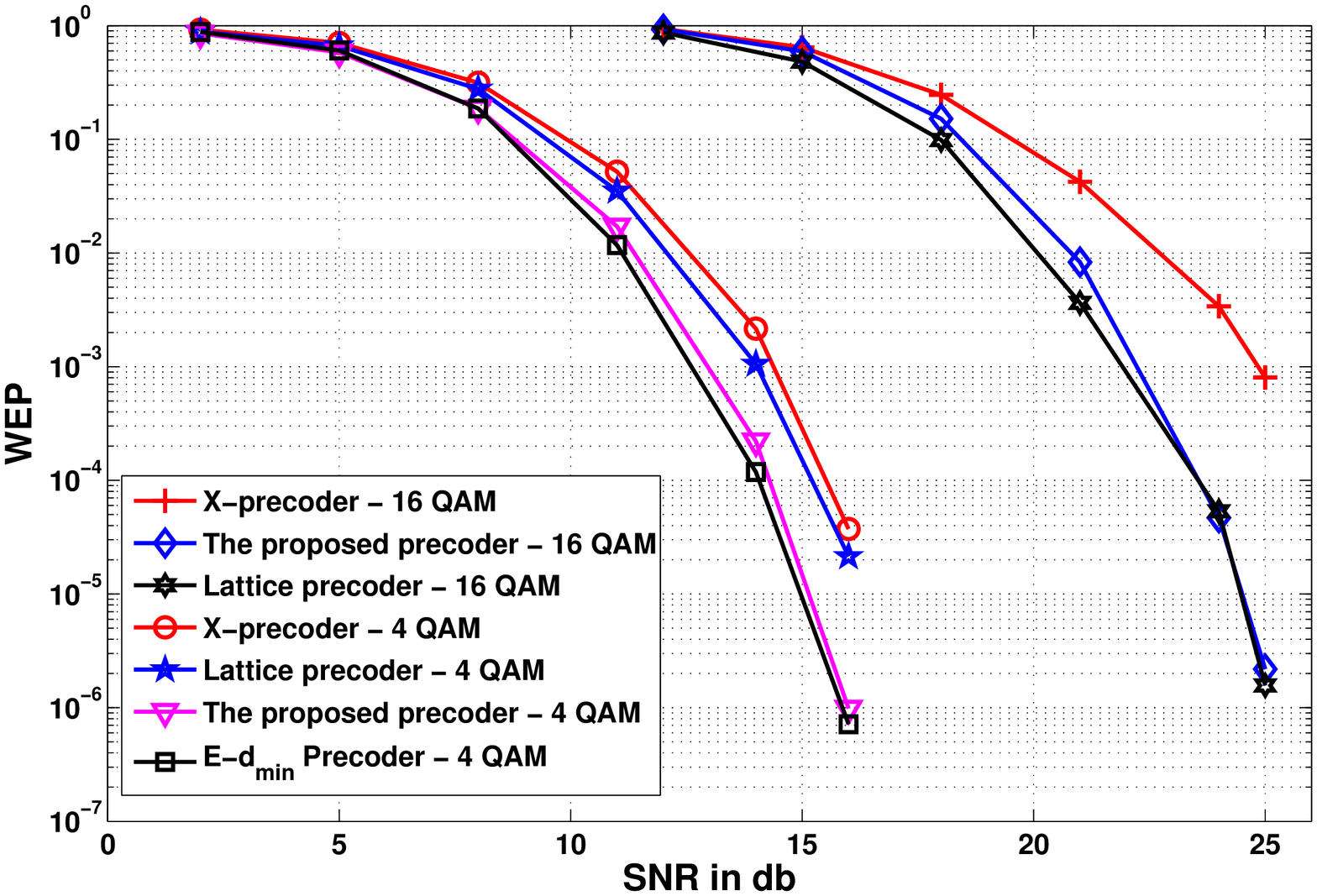}
\caption{WEP comparison for $8\times 8$ MIMO systems for $4$-QAM and $16$-QAM}
\label{fig_cer8x8}
\end{figure}

\end{document}